\documentclass{article}

\usepackage[a4paper,top=2cm,bottom=2cm,left=3cm,right=3cm,marginparwidth=1.75cm]{geometry}
\usepackage{amsmath}
\usepackage{amssymb}
\usepackage{amsthm}
\usepackage{graphicx}
\usepackage{hyperref}
\usepackage[all]{xy}
\usepackage{enumitem}
\usepackage{authblk}
\usepackage{color}

\theoremstyle{plain}
\newtheorem{theorem}{Theorem}[section]
\newtheorem{lemma}[theorem]{Lemma}

\newtheorem{corollary}[theorem]{Corollary}

\theoremstyle{definition}
\newtheorem{definition}[theorem]{Definition}
\newtheorem{example}[theorem]{Example}
\theoremstyle{remark}
\newtheorem{remark}[theorem]{Remark}

\newcommand{\Z}{\mathbb{Z}}
\newcommand{\R}{\mathbb{R}}
\newcommand{\RR}{\mathrm{R}}
\newcommand{\A}{\mathcal{A}}

\newcommand{\D}{\mathit{D}}
\renewcommand{\SS}{\mathcal{S}}
\newcommand{\col}{\mathrm{col}}

\newcommand{\fg}{\mathrm{fg}}
\newcommand{\incl}{\mathrm{incl}}
\newcommand{\SO}{\mathrm{SO}}

\renewcommand{\Im}{\operatorname{Im}}

\newcommand{\Map}{\mathrm{Map}}
\newcommand{\Hom}{\mathrm{Hom}}
\newcommand{\ang}[1]{\langle#1\rangle}

\title{Homotopy classification of knotted defects in ordered media}
\author[1,5]{Yuta Nozaki}
\author[2,5]{Tam\'{a}s K\'{a}lm\'{a}n}
\author[3,5]{Masakazu Teragaito}
\author[4,5]{Yuya Koda}
\affil[1]{\footnotesize Faculty of Environment and Information Sciences, Yokohama National University, Yokohama 240-8501, Japan;
\texttt{nozaki-yuta-vn@ynu.ac.jp}}
\affil[2]{\footnotesize Department of Mathematics, Tokyo Institute of Technology, 
Tokyo 152-8551, Japan;
\texttt{kalman@math.titech.ac.jp}}
\affil[3]{\footnotesize Department of Mathematics Education,
Hiroshima University, Higashi-Hiroshima 739-8524, Japan;
\texttt{teragai@hiroshima-u.ac.jp}}
\affil[4]{\footnotesize Department of Mathematics, Hiyoshi Campus, Keio University, Yokohama 223-8521, Japan;
\texttt{koda@keio.jp}}
\affil[5]{\footnotesize 
International Institute for Sustainability with Knotted Chiral Meta Matter (WPI-SKCM$^2$), Hiroshima University, Higashi-Hiroshima 739-0046, Japan.}

\begin{document}
\date{}
\maketitle

\begin{abstract}
We give a homotopy classification of the global defects in ordered media, and explain it 
via the example of biaxial nematic liquid crystals, i.e., systems where the order parameter space is the quotient of the $3$-sphere $S^3$ by the quaternion group $Q$.
As our mathematical model we consider continuous maps from
complements of spatial graphs to the space $S^3/Q$ modulo a certain equivalence relation, and find that the equivalence classes are enumerated by the six subgroups of $Q$.
Through monodromy around meridional loops, 
the edges of our spatial graphs are marked by conjugacy classes of $Q$; once we pass to planar diagrams, these labels can be refined to elements of $Q$ associated to each arc.
The same classification scheme applies not only in the case of $Q$ but also to arbitrary groups.
\end{abstract}

\section{Introduction}
Loops in a $3$-dimensional space can be knotted and linked in complicated ways.
Such objects are called knots or links in mathematics, more specifically in topology, and they have been extensively studied 
(see Rolfsen~\cite{Rol90} and Kawauchi~\cite{Kaw96}, for instance).
Since the discovery of the Jones polynomial \cite{Jon87} there has been a renewed attention to the relation between knot theory and theoretical physics, 
cf.\ Witten~\cite{Wit89}, Atiyah~\cite{Ati90}, and Kauffman~\cite{Kau13}.
But in fact, we often find knots in physical systems in more straightforward ways. Examples include
disclination lines of liquid crystals \cite{SLCT10,TRCZM11, SCZ14},
the cores of vortices in turbulent water \cite{KlIr13}
and superfluids \cite{KKI16},
as well as knottedness in optical and acoustic fields \cite{DKJOP10, KZLGZZZZ22,ZZLZLHZ20}.
Furthermore, knots appear as polymers \cite{SIT19},
strands of DNA \cite{HPLY10}, and
Skyrmion cores in classical field theory \cite{FaNi97}.
Note here that strands of polymers cannot cross, while vortex lines could cross without violating the topology. 

In this paper, we mainly focus on \emph{biaxial nematic liquid crystals} with $1$-dimensional defects. 
Here, a biaxial nematic liquid crystal is a material consisting of molecules that can be approximated by an ellipse, or some other two-dimensional shape with two discernible perpendicular symmetry axes.
The motion of the molecules is only constrained by the requirement that nearby principal axes (also called \emph{directors}) tend to align. It is well known and well observed that these systems can contain singularities in the form of one-dimensional loci which are characterized by non-trivial twisting around small meridional loops encircling them. That twisting prevents us from defining some of the director fields along the defect in a continuous manner. Because the directors are inequivalent, this can happen in several different ways which gives rise to various types of defects. The disclinations cannot terminate in the bulk but  it is possible for multiple defect lines to emanate from the same point. As long as such branching is at most three-fold, the resulting structure is robust in the sense that small perturbations will only give rise to small isotopies of the defect set.

Since the symmetries of an ellipse form the dihedral group $\D_2$ of order $4$, the positioning of a single molecule is parametrized by the quotient space $\SO(3)/\D_2$.
Note that $\D_2$ is isomorphic to $\Z_2\times\Z_2$, where $\Z_2$ is the cyclic group of order $2$.
Thus, a biaxial nematic liquid crystal texture (far from any boundary constraints) can be mathematically formulated as a continuous map $\R^3\setminus\Gamma \to \SO(3)/\D_2$, where the discontinuity $\Gamma$ is a spatial graph called the \emph{defect} of the texture.
In this context we refer to $\SO(3)/\D_2$ as the \emph{order parameter space} and we note that its fundamental group $\pi_1(\SO(3)/\D_2)$ is isomorphic to the \emph{quaternion group} $Q=\{\pm 1, \pm i, \pm j, \pm k\}$ (see Section~\ref{subsec:quaternion}).
A spatial graph is a finite, knotted graph in a $3$-dimensional space. Such objects have been studied in the context of knot theory and its applications (see Flapan~\cite{Fla00}, for instance).

There are many other kinds of physical systems described by an order parameter space. 
Local structures of line defects of such systems are completely classified by conjugacy classes (except that of the identity) of the fundamental group of the order parameter space.
We refer the reader to Po\'{e}naru--Toulouse~\cite{PoTo77}, Kl\'{e}man~\cite{Kle77}, Volovik--Mineev~\cite{VoMi77}, Mermin~\cite{Mer79}, Monastyrsky--Retakh~\cite{MoRe86}, Nakanishi--Hayashi--Mori~\cite{NHM88}, Brekke~et.~al.~\cite{BDHI92} and Masaki--Mizushima--Nitta~\cite{MMN24}.
See also the recent papers \cite{PBS20}, \cite{AlKa22}, and \cite{ZTWS23} on related topics. 
In our case, $Q$ is partitioned into the conjugacy classes $\{1\}$, $\{-1\}$, $\{\pm i\}$, $\{\pm j\}$, and $\{\pm k\}$. The latter four correspond to the four types of defects found in biaxial nematic liquid crystals. 

Let us now turn our attention to the global structure of the defect.
For a given link $L$, Machon and Alexander~\cite{MaAl14,MaAl16} investigated uniaxial nematic liquid crystals with defect $L$ and found that up to a natural equivalence they are enumerated by the so-called determinant of $L$.
Recently, Annala, Zamora-Zamora, and M\"{o}tt\"{o}nen~\cite{AZM22} considered interactions between defects called topologically allowed strand crossings and reconnections.
They classified biaxial nematic liquid crystal textures, with non-branching defects, up to these local moves.
For slightly different liquid crystals, Rajam\"{a}ki, Annala, and M\"{o}tt\"{o}nen~\cite{RAM23} obtained a partial result on the classification of defects.

In this paper we study, based on physical considerations (see, e.g., \cite{KKNU09}), 
certain local moves of defects, where as defects we allow not just links but also spatial graphs. 
In other words, we consider an equivalence relation among continuous maps $\R^3\setminus\Gamma \to \SO(3)/\D_2$, where $\Gamma$ runs over spatial graphs.
See Sections~\ref{sec:Mathematical model of global defects} and \ref{sec:Colored spatial graph diagrams} for a precise description.
We remark here that unfortunately, as far as we know, no truly knotted or linked defect lines have been observed experimentally in the actual biaxial nematic phase. 
Cholesteric liquid crystals do have observed links (see \cite{WiBo74}) but they are not defects but rather preimages of the Hopf fibration. 
Thus, everything addressed in this paper is still somewhat theoretical.

\begin{theorem}
\label{thm:top_classif_Q}
Up to the equivalence relation, there are exactly six biaxial nematic liquid crystal textures in $\R^3$ whose defects are \textup{(}possibly empty\textup{)} spatial graphs. They correspond to the six subgroups of $Q$ \textup{(}see Figure~\ref{fig:classification_correspondence}\textup{)}.
\end{theorem}

\begin{figure}[htbp]
\centering\includegraphics[width=14cm]{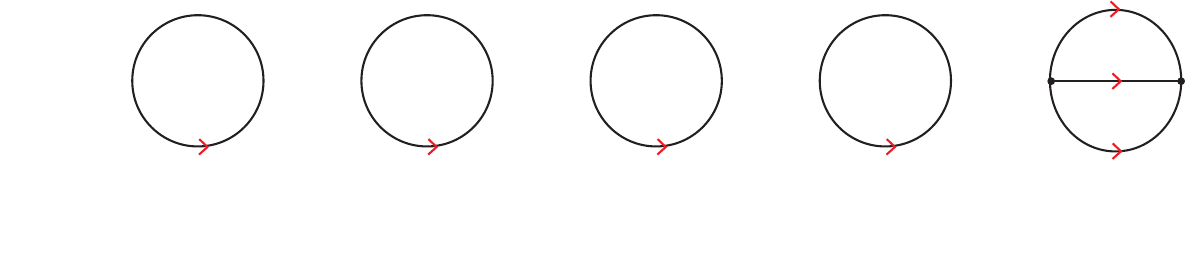}
\begin{picture}(400,0)(0,0)
\put(59,40){\color{red} $-1$}
\put(143,40){\color{red} $i$}
\put(220,40){\color{red} $j$}
\put(297,40){\color{red} $k$}
\put(374,104){\color{red} $i$}
\put(374,81){\color{red} $j$}
\put(374,38){\color{red} $k$}

\put(0,70){$\emptyset$}

\put(0,22){$\updownarrow$}
\put(63,22){$\updownarrow$}
\put(142,22){$\updownarrow$}
\put(218,22){$\updownarrow$}
\put(297,22){$\updownarrow$}
\put(374,22){$\updownarrow$}

\put(-05,0){$\{1\}$}
\put(53,0){$\{\pm 1\}$}
\put(126,0) {$\{\pm 1, \pm i\}$}
\put(202,0){$\{\pm 1, \pm j\}$}
\put(281,0){$\{\pm 1,\pm k\}$}
\put(373,0){$Q$}
\end{picture}
\caption{Representatives of each
class 
of 
global defects in biaxial nematic liquid crystals in $\R^3$ or $S^3$, up to equivalence. Labels (with directions) indicate local structure
(see Section~\ref{sec:Colored spatial graph diagrams}). 
}
\label{fig:classification_correspondence}
\end{figure}
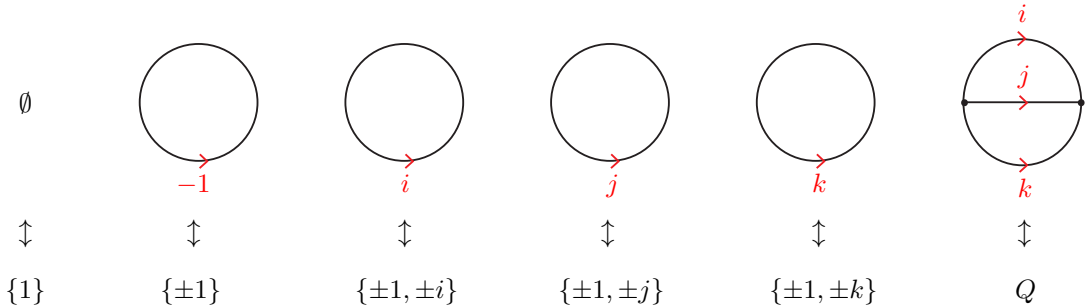

In parts of this work we also consider textures that are uniform far from the origin, or equivalently, those in the $3$-dimensional sphere $S^3$ (see Remark~\ref{rem:S3}).
Our framework can be naturally generalized from the quaternion group $Q$ to an arbitrary group $G$. 
As an additional example, the case of the tetrahedral group will be discussed in Appendix~\ref{sec:Tetrahedral}.
Then, Theorem~\ref{thm:top_classif_Q} becomes a special case of Theorem~\ref{thm:top_classif}.
To prove Theorem~\ref{thm:top_classif}, we introduce purely combinatorial objects called \emph{$G$-colored spatial graph diagrams} and local moves among them (see Section~\ref{sec:Colored spatial graph diagrams} for the definition).

\subsection*{Acknowledgments}
The authors 
thank Ivan I. Smalyukh, Muneto Nitta, Michikazu Kobayashi, and David Palmer 
for many helpful comments.
They also thank the anonymous referees for their valuable  suggestions.
This study was supported in part by JSPS KAKENHI Grant Numbers JP20K14317, JP23K03108, JP23K12974, JP20K03588, JP21H00978, JP23H05437, and JP24K06744.

\section{Preliminaries}
\label{sec:preliminaries}
In this section, we first briefly review the necessary basics of groups and topology. 
See for example Hatcher~\cite{Hat02} and Nakahara~\cite{Nak03} for further details. 

\subsection{Quaternion group}
\label{subsec:quaternion}

Let $G$ be a group. 
We say that $G$ has a \emph{presentation} 
$\langle\, S \mid R \,\rangle$, where $S$ and $R$ are called 
sets of \emph{generators} and \emph{relators}, respectively, if each element of $G$ can be written as a product of powers of elements of $S$, and $R$ is a set of products of generators, each equal to the unity, whose products are sufficient to express all relations among the generators. 
More formally, $G$ is isomorphic to 
the quotient group of the free group $F(S)$ on 
$S$ by the normal closure of the subset $R$ of $F(S)$. 
A group $G$ is said to be \emph{finitely generated} 
if $G$ has a presentation $\langle\, S \mid R \,\rangle$ with 
a finite $S$. 

The \emph{quaternion group} 
$Q=\{\pm 1, \pm i, \pm j, \pm k\}$ is the non-abelian group generated by four elements $-1, i, j, k$ such that 
$-1$ commutes with other elements and 
they satisfy $(-1)^2 =1$ and $i^2 = j^2 = k^2 = ijk = -1$. 
Thus, a (non-minimal) presentation of $Q$ is 
\[ 
\langle\, \varepsilon, i,j,k \mid 
\varepsilon i \varepsilon^{-1}i^{-1},\ 
\varepsilon j \varepsilon^{-1}j^{-1},\ 
\varepsilon k \varepsilon^{-1}k^{-1},\
\varepsilon^2,\ 
\varepsilon i^2,\ \varepsilon j^2,\ \varepsilon k^2,\ 
\varepsilon ijk
\,\rangle, 
\]
where $\varepsilon$ corresponds to $-1$. 

Two elements $x$ and $y$ of a group $G$ are said to be \emph{conjugate} if there exists an element $a \in G$ 
with $y = a x a^{-1}$. 
Similarly, two subgroups $H$ and $H'$ of $G$ 
are \emph{conjugate} if there exists an element $a \in G$ 
with $H' = a H a^{-1}$. 
These are equivalence relations on $G$ itself and 
the set of subgroups of $G$, respectively, 
and in both cases the equivalence classes are called \emph{conjugacy classes}. 
Note that a normal subgroup is nothing but a subgroup that is conjugate only to itself. 
In abelian groups, 
distinct elements can never be conjugate, and 
every subgroup is normal. 
The quaternion group $Q$ has five conjugacy classes of elements: 
\[
[1]=\{1\},\ [-1] = \{-1\},\  
[i] = \{ i, -i \},\ 
[j] = \{ j, -j \},\text{ and }
[k] = \{ k, -k \}.
\]
In $Q$ there exist six subgroups in total, namely
\[ \langle 1 \rangle = \{1\},\ 
\langle -1 \rangle = \{ 1, -1 \},\ 
\langle i \rangle = \{ \pm 1, \pm i \},\ 
\langle j \rangle = \{ \pm 1, \pm j \},\ 
\langle k \rangle = \{ \pm 1, \pm k \},\ Q,
\]
and neither is conjugate to any other 
just like in the case of an abelian group.

\subsection{Fundamental group and Wirtinger presentation}
\label{subsec:Wirtinger}
Let $X$ be a topological space with a basepoint $x_0$ and let $p_0$ be a basepoint of the circle $S^1$.
We write $\pi_1(X,x_0)$ for the set of continuous maps $(S^1,p_0)\to (X,x_0)$ up to basepoint-preserving homotopy.
One can define the multiplication of two loops by their concatenation, which is well-defined up to homotopy.
This multiplication turns $\pi_1(X,x_0)$ into a group, which is called the \emph{fundamental group} of $X$.
Throughout this paper, we simply write $\pi_1(X)$ since $\pi_1(X,x_0)\cong \pi_1(X,x'_0)$ if $X$ is path-connected.
For instance, $\pi_1(S^1)\cong \Z$, $\pi_1(S^3)\cong \{1\}$, and $\pi_1(\SO(3))\cong \Z_2$.

\begin{figure}[htbp]
\centering\includegraphics[width=13cm]{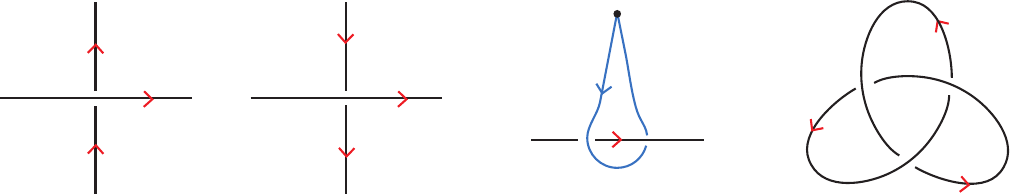}
\begin{picture}(400,0)(0,0)

\put(65,36){\color{red} $x_i$}
\put(35,24){\color{red} $x_k$}
\put(35,64){\color{red} $x_j$}

\put(158,36){\color{red} $x_i$}
\put(127,24){\color{red} $x_k$}
\put(127,64){\color{red} $x_j$}

\put(245,76){basepoint}

\put(365,72){\color{red} $x_1$}
\put(298,36){\color{red} $x_2$}
\put(368,6){\color{red} $x_3$}
\end{picture}
\caption{A positive crossing, negative crossing, the loop that corresponds to an arc, and a diagram of the trefoil knot.}
\label{fig:Wirtinger}
\end{figure}
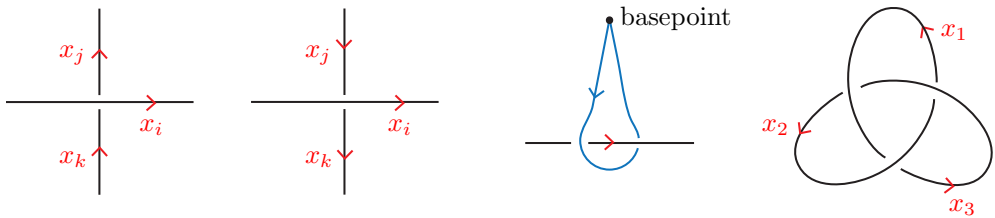

For a link $L$ in $S^3$ (or $\R^3$), let $E(L)$ denote the complement of an open tubular neighborhood of $L$.
The fundamental group $\pi_1(E(L))$ can be explicitly described from a diagram of $L$.
Let $D$ be a diagram (i.e., regular projection) of the link $L$ with $m$ crossings and $n$ arcs, each of which is either a loop without self-crossings, or connects two undercrossing sites.
By assigning letters $x_1,\dots,x_n$ to the arcs, we obtain a presentation for $\pi_1(E(L))$ with $n$ generators $x_1,\dots,x_n$ and $m$ relators. Namely, if we place the basepoint ``above'' the projection plane then the generator $x_i$ is thought of as a loop that descends from there, links the arc $x_i$ once as shown in Figure \ref{fig:Wirtinger}, and returns to the basepoint without interacting with any other part of $L$.
Furthermore, each relator comes from a crossing and is of the form $x_ix_jx_i^{-1}x_k^{-1}$ (resp.\ $x_ix_j^{-1}x_i^{-1}x_k$) if the corresponding crossing is positive (resp.\ negative), as drawn in Figure~\ref{fig:Wirtinger}.
This presentation is called the \emph{Wirtinger presentation} of $\pi_1(E(L))$.
For instance, we obtain
\[
\pi_1(E(K))\cong \ang{x_1,x_2,x_3\mid x_1x_3^{-1}x_1^{-1}x_2,\ x_2x_1^{-1}x_2^{-1}x_3,\ x_3x_2^{-1}x_3^{-1}x_1}
\cong \ang{x_1,x_2\mid x_1x_2x_1^{-1}x_2^{-1}x_1^{-1}x_2}
\]
for the trefoil knot $K$ illustrated on the right in Figure~\ref{fig:Wirtinger}.

For pairs of topological spaces $(X,A)$ and $(Y,B)$, let $\Map((X,A),(Y,B))$ denote the set of continuous maps $f\colon X\to Y$ satisfying $f(A)\subset B$.
In particular, the set $\Map((X,\{x_0\}),(Y,\{y_0\}))$ of basepoint-preserving continuous maps is simply denoted by $\Map_0(X,Y)$.
There is a natural correspondence $\Map_0(X,Y)\to \Hom(\pi_1(X),\pi_1(Y))$, defined by sending a map $f\colon X \to Y$ to the induced homomorphism between the fundamental groups. 
This gives rise to another natural mapping
\[\Map_0(X,Y)/{\simeq}\to \Hom(\pi_1(X),\pi_1(Y)),\]
where $\simeq$ is the equivalence relation of basepoint-preserving homotopy.
Here $\Hom(G,H)$ denotes the set of homomorphisms from a group $G$ to $H$.

For an integer $n\geq 2$, the $n$th homotopy group $\pi_n(X,x_0)$ is obtained from $\Map_0(S^n,X)$ by considering basepoint-preserving homotopies and a concatenation operation similar to that of $\pi_1$. In contrast to $\pi_1$, these groups are always abelian.

\section{Mathematical model of global defects}
\label{sec:Mathematical model of global defects}

Recall that the order parameter space for biaxial nematic liquid crystals is $\SO(3)/\D_2$ and their global defects can be regarded as spatial graphs $\Gamma$ in $\R^3$.
Mathematically, they correspond to continuous maps $\R^3\setminus\Gamma \to \SO(3)/\D_2$.
Note here that $\pi_1(\SO(3)/\D_2)\cong Q$ and $\pi_2(\SO(3)/\D_2)\cong \{0\}$ since $\pi_2(\SO(3))=\{0\}$ and the $n$th homotopy groups of covering spaces agree for each $n\geq2$.

We now generalize this situation to an arbitrary group $G$.
Let $X_G$ be a CW-complex satisfying $\pi_1(X_G)\cong G$ and $\pi_2(X_G)=\{0\}$.
(Such spaces exist for all $G$.)
The next lemma, used in Section~\ref{sec:Classification}, allows us to describe continuous maps in terms of group homomorphisms.

\begin{lemma}
\label{lem:2-skeleton}
The natural correspondence $\Map_0(\R^3\setminus\Gamma, X_G)/{\simeq}\to \Hom(\pi_1(\R^3\setminus\Gamma), G)$ is a bijection, where $\simeq$ is the equivalence relation of basepoint-preserving homotopy.
\end{lemma}

\begin{proof}
First note that $\R^3\setminus\Gamma$ is homotopy equivalent to a CW-complex $Y$ of dimension at most $2$ and let $Y^{(d)}$ denote the $d$-skeleton of $Y$.
It suffices to give the inverse map of $\phi\colon \Map_0(Y, X_G)/{\simeq}\to \Hom(\pi_1(Y), G)$.
Let $f\in \Hom(\pi_1(Y), G)$.
Then $f$ is induced by some continuous map $Y^{(1)}\to X_G$.
Since $f$ is a homomorphism between fundamental groups, the map extends to $Y^{(2)}=Y \to X_G$, which is unique up to homotopy because $\pi_2(X_G)=\{0\}$.
This gives the inverse map of $\phi$.
\end{proof}

\begin{definition}
\label{def:equivalence}
Fix a basepoint $p_0$ in $\R^3$.
We define an equivalence relation $\sim$ on the disjoint union $\coprod_\Gamma \Map_0(\R^3\setminus \Gamma, X_G)$, where $\Gamma$ runs over all spatial graphs (i.e., $1$-dimensional finite CW-complexes) in $\R^3\setminus\{p_0\}$.
Two maps $f\in \Map_0(\R^3\setminus \Gamma, X_G)$ and $g\in \Map_0(\R^3\setminus \Gamma', X_G)$ are \emph{equivalent}, $f\sim g$, if there exists a basepoint-preserving ambient isotopy $h_t$ of $\R^3$ sending $\widehat{\Gamma}=\Gamma\cup\{v_i\}\cup\{e_j\}$ to 
$\widehat{\Gamma}'=\Gamma'\cup\{v'_k\}\cup\{e'_l\}$ in such a way that $f|_{\R^3\setminus \widehat{\Gamma}}$ is homotopic to $g|_{\R^3\setminus \widehat{\Gamma}'}\circ h_1$, after adding suitable vertices $v_i$, $v'_k$ and edges $e_j$, $e'_l$ to each graph.
\end{definition}

This equivalence relation is one possible mathematical model for local deformation of defects.
Here we remark that, in our model, 
the following three types of moves, and their inverses, cannot be realized within the same equivalence class in general:
\begin{enumerate}[label=(\arabic*)]
\item\label{ichi}
shrinking an unknotted defect loop 
to a point
\item \label{ni}
splitting a segment of defect into two parallel copies, with the appropriate labels, joining at two new vertices of degree $3$
\item\label{san}
splitting a defect loop into two parallel components.
\end{enumerate}
We note that allowing \ref{ni} would imply that \ref{san} is possible as well. 
But in fact, there are some physical considerations against doing so.
As to \ref{ichi}, it is mentioned in \cite{AZM22} that
in 
the Bose--Einstein condensates discussed in \cite{RuAn03}, 
an energy barrier prevents an unknotted defect loop from contracting to a point defect.
Furthermore,  
the splitting of defects 
is energetically unfavorable in some biaxial nematic liquid crystal textures, 
and it is not observed in simulations (see \cite{PrPe02}).

Mathematically, it would not be hard to incorporate the above 
moves into the model. 
However, if we included any one of the three, that would be enough to make any defect configuration equivalent to any other, that is, the classification would become trivial.

Returning to our model above, we can define a map
\[
\Phi\colon \left(\coprod_\Gamma \Map_0(\R^3\setminus \Gamma, X_G)\right)/{\sim} \to \SS_G^\fg
\]
by $\Phi([f])=\Im(f_\ast\colon \pi_1(\R^3\setminus\Gamma)\to G)$, where $\SS_G^\fg$ denotes the set of finitely generated subgroups of $G$.
Indeed if $f\sim g$ via $h_t$ then we have a commutative diagram
\[
\xymatrix{
\pi_1(\R^3\setminus \widehat{\Gamma}) \ar@{->>}[r]^-{\incl_\ast} \ar[d]_-{h_1}^-{\cong} & \pi_1(\R^3\setminus \Gamma) \ar[r]^-{f_\ast} & \pi_1(X_G) \ar@{=}[d] \\
\pi_1(\R^3\setminus \widehat{\Gamma}') \ar@{->>}[r]^-{\incl_\ast} & \pi_1(\R^3\setminus \Gamma') \ar[r]^-{g_\ast} & \pi_1(X_G),
}
\]
and thus $\Phi$ is well-defined.
The main result of this paper is the following theorem, which gives the complete classification of continuous maps up to the above equivalence relation.
In particular, the case $G=Q$ provides the complete classification of biaxial nematic liquid crystals in our model.
As to the proof, we will see that our claim is a
direct consequence of Theorem~\ref{thm:comb_classif} below.

\begin{theorem}
\label{thm:top_classif}
The map $\Phi$ is a bijection.
\end{theorem}

\begin{remark}
When we describe biaxial nematic liquid crystals in an orientable $3$-manifold $M$ in terms of continuous maps, we need to fix a framing on $M$, that is, a trivialization of the tangent bundle of $M$.
In the case $M$ is $\R^3$ or one of its subspaces, one can use the standard framing of $\R^3$.
\end{remark}

\section{Colored spatial graph diagrams}
\label{sec:Colored spatial graph diagrams}
In this section we introduce $G$-colored diagrams, which enable us to manipulate continuous maps in a combinatorial way.

A \emph{spatial graph} is a finite $1$-dimensional CW-complex in $\R^3$ or $S^3$, which is allowed to be disconnected and possibly has multiple edges, loops, and isolated points.
An oriented spatial graph is a spatial graph whose edges are oriented.
For an (oriented) spatial graph $\Gamma$, we always have a regular projection $D\subset \R^2$ of $\Gamma$, that is, a projection that restricts to a homeomorphism between $\Gamma$ and $D$ except for finitely many transversal double points in $D$.
At each double point, 
the relative heights of its preimages are indicated by drawing the ``undercrossing'' strand with a short gap.
In the resulting diagram $D$, let us write $\A_D$ for the set of arcs. In this case arcs are either free loops (lying `above' other parts of $D$) or curved, embedded segments whose endpoints are either undercrossings or vertices of $D$. 

Let $G$ be a group, for instance, the quaternion group $Q$.

\begin{definition}
\label{def:diagram}
A \emph{$G$-colored oriented spatial graph diagram} (or simply \emph{$G$-colored diagram}) $D$ is a regular projection of an oriented spatial graph such that each arc of $D$ 
has an associated 
element of $G$, and these satisfy relations arising from crossings and vertices.
Namely, at each crossing the corresponding relator is the same as in the Wirtinger presentation, and at a vertex of degree $k$, as drawn in Figure~\ref{fig:vertex_relation}, the corresponding relator is $c_1^{\varepsilon_1}\cdots c_k^{\varepsilon_k}$, where $\varepsilon_i=1$ (resp.\  $\varepsilon_i=-1$) if the $i$th half-edge is a tail (resp.\ head).
Finally, the color of $\alpha\in \A_D$ will be denoted by $\col(\alpha)$.
\end{definition}

\begin{figure}[htbp]
\centering\includegraphics[width=9cm]{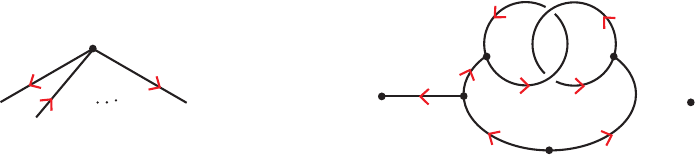}
\begin{picture}(400,0)(0,0)
\put(76,44){\color{red} $c_1$}
\put(90,22){\color{red} $c_2$}
\put(130,42){\color{red} $c_k$}

\put(290,10){\color{red} $c_1$}
\put(290,70){\color{red} $c_2$}
\put(282,28){\color{red} $c_3$}
\put(248,70){\color{red} $c_4$}
\put(258,28){\color{red} $c_5$}
\put(233,45){\color{red} $c_6$}
\put(224,24){\color{red} $c_7$}
\put(248,10){\color{red} $c_8$}

\end{picture}
\caption{Half-edges around a vertex and an example of a $G$-colored diagram.}
\label{fig:vertex_relation}
\end{figure}
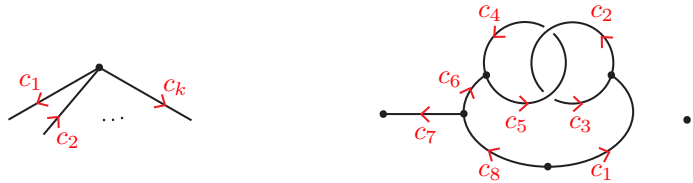

Note that we sometimes omit vertices of degree two in $G$-colored diagrams of knots and links, as in  Figure~\ref{fig:classification_correspondence}.
We stress that each color $\col(\alpha)$ is an element and not a conjugacy class of $G$.
Among other things, this is necessary for the relators above to make sense.
A $G$-colored diagram $D$ is said to be \emph{equivalent} to $D'$ if $D'$ is obtained from $D$ by a finite sequence of the following local moves and ambient isotopy of $\R^2$.

\begin{enumerate}[label=(\arabic*)]
    \item 
    Consider a path in a $G$-colored diagram corresponding to 
    an edge of the spatial graph. 
    Then, a move that reverses the orientations and replaces the colors 
    with their inverses, for all arcs along the path simultaneously, is called an \emph{orientation-reversal}. 
    See Figure \ref{fig:local_move_1}. 
    This move changes some orientations and colors but 
    does not change the diagram. 
    
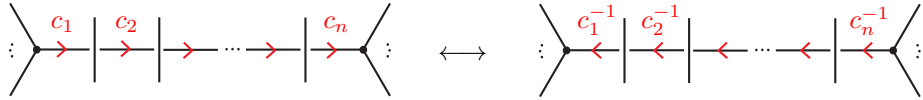
\begin{figure}[htbp]
\centering\includegraphics[width=12cm]{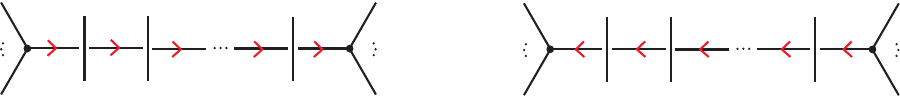}
\begin{picture}(400,0)(0,0)

\put(45,38){\color{red} $c_1$}
\put(69,38){\color{red} $c_2$}
\put(147,38){\color{red} $c_n$}

\put(242,38){\color{red} $c_1^{-1}$}
\put(266,38){\color{red} $c_2^{-1}$}
\put(344,38){\color{red} $c_n^{-1}$}

\put(190,26){$\longleftrightarrow$}

\end{picture}
\caption{Orientation-reversal.}
\label{fig:local_move_1}
\end{figure}
    \item 
    A \emph{Reidemeister move} is a one of the five local moves 
    $\RR_1$--$\RR_5$ shown in Figure \ref{fig:local_move_2}. 
    These moves change diagrams, but do not change the isotopy class of the corresponding spatial graph. 
    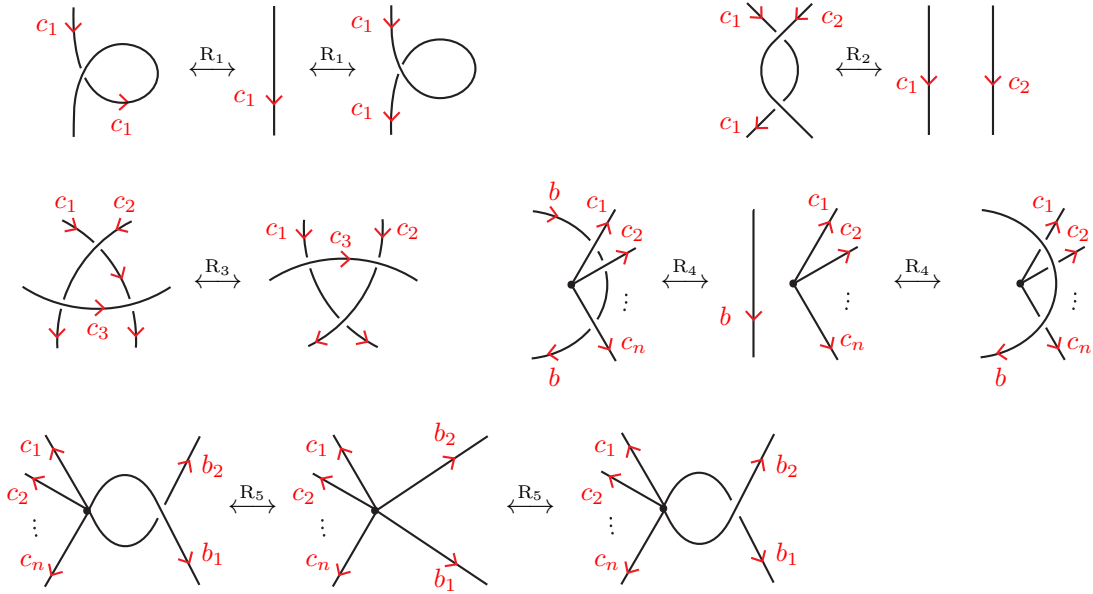
\begin{figure}[htbp]
\centering\includegraphics[width=14cm]{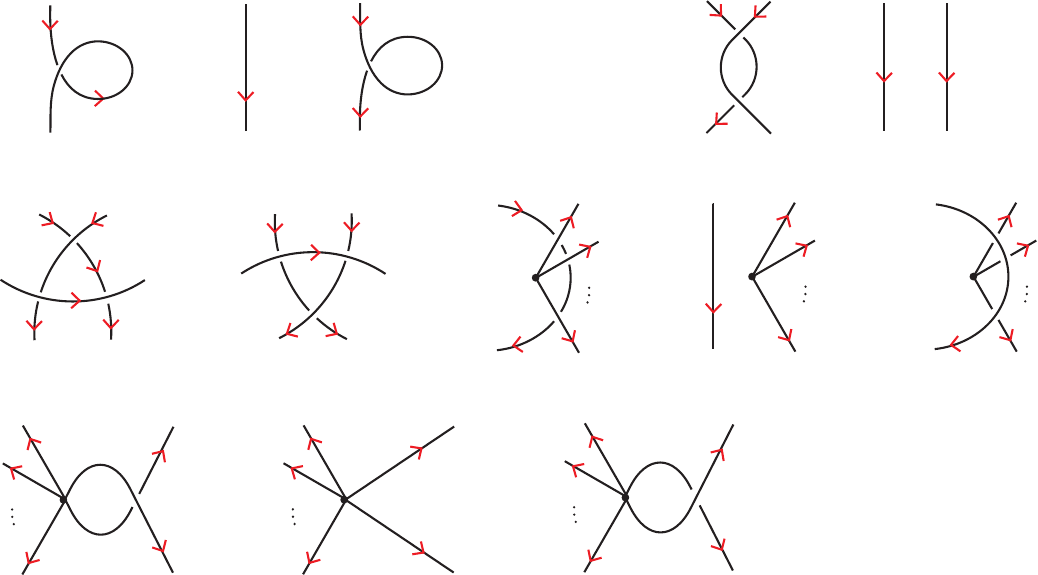}
\begin{picture}(400,0)(0,0)

\put(6,222){\color{red} $c_1$}
\put(34,184){\color{red} $c_1$}
\put(80,195){\color{red} $c_1$}
\put(124,190){\color{red} $c_1$}
\put(124,223){\color{red} $c_1$}
\put(63,205){$\overset{\RR_1}{\longleftrightarrow}$}
\put(108,205){$\overset{\RR_1}{\longleftrightarrow}$}

\put(262,225){\color{red} $c_1$}
\put(300,225){\color{red} $c_2$}
\put(262,185){\color{red} $c_1$}
\put(305,205){$\overset{\RR_2}{\longleftrightarrow}$}
\put(328,200){\color{red} $c_1$}
\put(370,200){\color{red} $c_2$}

\put(13,155){\color{red} $c_1$}
\put(35,155){\color{red} $c_2$}
\put(25,107){\color{red} $c_3$}
\put(65,125){$\overset{\RR_3}{\longleftrightarrow}$}
\put(92,145){\color{red} $c_1$}
\put(141,145){\color{red} $c_2$}
\put(116,142){\color{red} $c_3$}

\put(198,158){\color{red} $b$}
\put(198,88){\color{red} $b$}
\put(212,155){\color{red} $c_1$}
\put(225,144){\color{red} $c_2$}
\put(225,102){\color{red} $c_n$}
\put(240,125){$\overset{\RR_4}{\longleftrightarrow}$}
\put(262,111){\color{red} $b$}
\put(295,156){\color{red} $c_1$}
\put(307,145){\color{red} $c_2$}
\put(307,103){\color{red} $c_n$}
\put(327,125){$\overset{\RR_4}{\longleftrightarrow}$}
\put(365,88){\color{red} $b$}
\put(379,155){\color{red} $c_1$}
\put(392,144){\color{red} $c_2$}
\put(392,102){\color{red} $c_n$}

\put(0,63){\color{red} $c_1$}
\put(-5,45){\color{red} $c_2$}
\put(0,20){\color{red} $c_n$}
\put(68,55){\color{red} $b_2$}
\put(68,22){\color{red} $b_1$}
\put(78,40){$\overset{\RR_5}{\longleftrightarrow}$}
\put(107,63){\color{red} $c_1$}
\put(102,45){\color{red} $c_2$}
\put(107,20){\color{red} $c_n$}
\put(155,67){\color{red} $b_2$}
\put(155,12){\color{red} $b_1$}
\put(182,40){$\overset{\RR_5}{\longleftrightarrow}$}
\put(215,63){\color{red} $c_1$}
\put(210,45){\color{red} $c_2$}
\put(215,20){\color{red} $c_n$}
\put(283,55){\color{red} $b_2$}
\put(283,22){\color{red} $b_1$}
\end{picture}
\caption{Reidemeister moves: The colors of the short `local' arcs are omitted because they are uniquely determined by the colors of the other arcs, according to the 
appropriate Wirtinger relators.
}
\label{fig:local_move_2}
\end{figure}
    \item 
    A move that eliminates or 
    adds a vertex of degree $2$, as shown on the left side in Figure~\ref{fig:local_moves_3-4}, is called an  
    \emph{edge-combining} or \emph{edge-subdivision}. 
    This move changes the vertex and edge sets 
    of the corresponding spatial graph, but 
    does not change the underlying subset of $\R^3$.
    Even though this move is a special case of \ref{yon}, we include it separately for convenience.
\begin{figure}[htbp]
\centering\includegraphics[width=15cm]{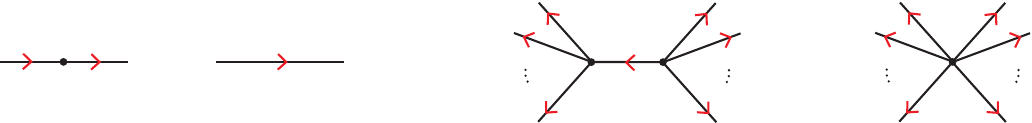}
\begin{picture}(400,0)(0,0)
\put(-5,43){\color{red} $c$}
\put(23,43){\color{red} $c$}
\put(98,43){\color{red} $c$}
\put(48,35){$\longleftrightarrow$}

\put(214,62){\color{red} $b_1$}
\put(194,39){\color{red} $b_2$}
\put(198,18){\color{red} $b_m$}

\put(231,27){\color{red} $b_1 \cdots b_m$}

\put(268,64){\color{red} $c_1$}
\put(288,39){\color{red} $c_2$}
\put(283,20){\color{red} $c_n$}

\put(309,35){$\longleftrightarrow$}

\put(366,62){\color{red} $b_1$}
\put(346,39){\color{red} $b_2$}
\put(348,18){\color{red} $b_m$}
\put(389,62){\color{red} $c_1$}
\put(408,39){\color{red} $c_2$}
\put(404,18){\color{red} $c_n$}

\end{picture}
\caption{Edge-combining and edge-subdivision (left);
edge-contraction and vertex-splitting (right).}
\label{fig:local_moves_3-4}
\end{figure}
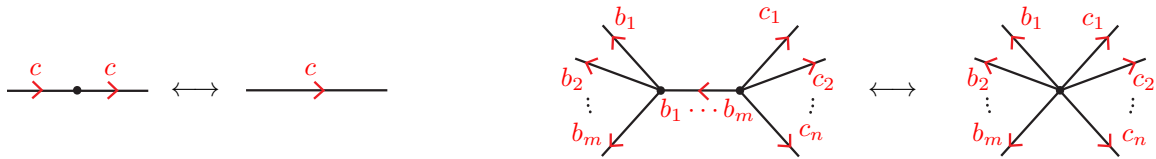

    \item \label{yon}
    A move that contracts an edge of a $G$-colored diagram to 
    a point is called an \emph{edge-contraction}. 
    The inverse of this move, that is, 
    a move that divides the half-edges incident to a prefixed vertex into two sets, makes each set incident to a separate vertex, and finally joins those two vertices by an edge, where 
    the color of that new edge is determined naturally from those of the half-edges, is called a \emph{vertex-splitting}.
    See the right side of Figure~\ref{fig:local_moves_3-4}. 
    Note that in the figure, the color of the newly added edge 
    under the vertex-splitting is $b_1\cdots b_m$, which is equal to 
    $c_1^{-1}\cdots c_n^{-1}$. 
    This move changes the corresponding spatial graph, but 
    does not change the regular neighborhood of the spatial graph in space. 

%
%
%
%
    \item 
    A move that adds or removes an edge of the color $1$ 
    is called an \emph{edge-addition} or \emph{edge-deletion}. 
    See Figure \ref{fig:local_move_5}. 
    This move does change the corresponding spatial graph, and
    even changes the isotopy class of its regular neighborhood. 
\begin{figure}[htbp]
\centering\includegraphics[width=8cm]{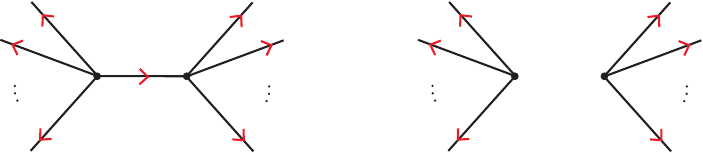}
\begin{picture}(400,0)(0,0)
\put(102,62){\color{red} $b_1$}
\put(82,37){\color{red} $b_2$}
\put(84,18){\color{red} $b_m$}
\put(157,62){\color{red} $c_1$}
\put(172,37){\color{red} $c_2$}
\put(169,18){\color{red} $c_n$}

\put(188,35){$\longleftrightarrow$}

\put(238,62){\color{red} $b_1$}
\put(217,37){\color{red} $b_2$}
\put(220,18){\color{red} $b_m$}
\put(130,24){\color{red} $1$}
\put(290,62){\color{red} $c_1$}
\put(305,37){\color{red} $c_2$}
\put(305,18){\color{red} $c_n$}
\end{picture}
\caption{Edge-addition and edge-deletion.}
\label{fig:local_move_5}
\end{figure}
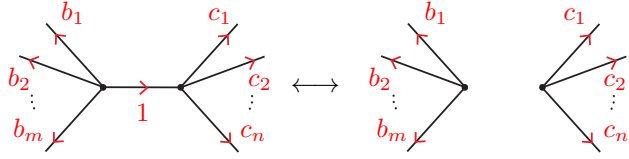
    \item 
    A move that adds or removes an isolated vertex 
    is called a \emph{vertex-addition} or \emph{vertex-deletion}. 
    This also changes the spatial graph, as well as the isotopy class of its regular neighborhood. 
    \item 
    For an element $x \in G$, a move replacing 
    $\col(\alpha)$ with $x\col(\alpha)x^{-1}$ for every $\alpha\in \A_D$ is called a 
    \emph{simultaneous conjugation}. 
\end{enumerate}



\begin{example}
\label{ex:crossing_change}
Suppose a given colored diagram contains a part where two edges of the graph are locally hooked to each other as shown 
on the left of 
Figure~\ref{fig:crossing_change}. 
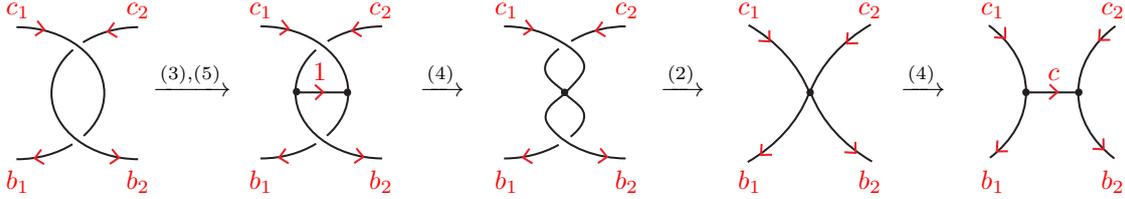
\begin{figure}[h!]
\centering\includegraphics[width=14.5cm]{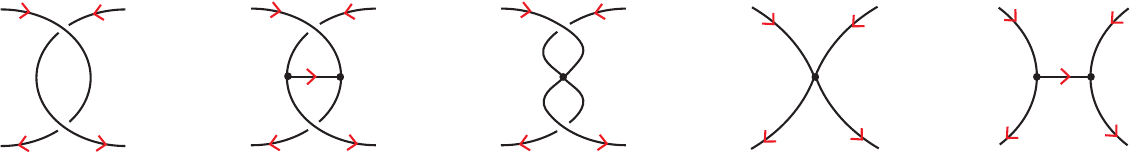}
\begin{picture}(400,0)(0,0)
\put(-10,69){\color{red}$c_1$}
\put(35,69){\color{red}$c_2$}
\put(-10,3){\color{red}$b_1$}
\put(35,3){\color{red}$b_2$}

\put(45,38){$\xrightarrow{(3),(5)}$}

\put(81,69){\color{red}$c_1$}
\put(126,69){\color{red}$c_2$}
\put(81,3){\color{red}$b_1$}
\put(126,3){\color{red}$b_2$}
\put(105,44){\color{red}$1$}

\put(145,38){$\xrightarrow{(4)}$}

\put(173,69){\color{red}$c_1$}
\put(218,69){\color{red}$c_2$}
\put(173,3){\color{red}$b_1$}
\put(218,3){\color{red}$b_2$}

\put(235,38){$\xrightarrow{(2)}$}

\put(264,69){\color{red}$c_1$}
\put(309,69){\color{red}$c_2$}
\put(264,3){\color{red}$b_1$}
\put(309,3){\color{red}$b_2$}

\put(325,38){$\xrightarrow{(4)}$}

\put(355,69){\color{red}$c_1$}
\put(400,69){\color{red}$c_2$}
\put(355,3){\color{red}$b_1$}
\put(400,3){\color{red}$b_2$}
\put(380,44){\color{red}$c$}

\end{picture}
\caption{Sequence of moves that resolves a part in  a colored diagram where two edges 
are locally hooked to each other. 
In this figure, $b_1 = c_1 c_2 c_1 c_2^{-1} c_1^{-1}$, $b_2 = c_1 c_2 c_1^{-1}$, 
and $c = c_1 c_2 c_1^{-1} c_2^{-1}$.}
\label{fig:crossing_change}
\end{figure}
Then, by successively applying the moves introduced above as indicated in 
Figure~\ref{fig:crossing_change}, 
it is possible to resolve the hook. 
Here, the edge with the color $c$ in the 
figure is sometimes called a \textit{rung vortex}. 
We can easily check that $c$ is equal to 
$c_1 c_2 c_1^{-1} c_2^{-1}$, the so-called commutator of $c_1$ and $c_2$. 
Thus, in the case when $c_1$ and $c_2$ are commutative, $c$ equals $1$. This implies that 
we can formally eliminate
that edge by the move $(5)$. 
In other words, we can resolve the two crossings 
shown on the left of
Figure~\ref{fig:crossing_change} without creating a rung vortex in that case.
\end{example}

\begin{remark}
Similar to the steps of Example \ref{ex:crossing_change}, which produce a crossing change between strands of commutative colors, it is also easy to realize a reconnection of strands of the same color.
Here some of the intermediate stages are graphs, yet this is clearly enough to conclude that if two links are equivalent in the sense used in \cite{AZM22} then they are also equivalent, up to the moves (1)--(6) above, as spatial graphs.
\end{remark}

\section{Classification}
\label{sec:Classification}
One can define a map $\Psi\colon \{\text{$G$-colored diagrams}\}/\text{(1)--(6)} \to \SS_G^\fg$ by $\Psi(D)=\ang{\{\col(\alpha)\mid \alpha\in \A_D\}}$, where $\col(\alpha)$ denotes the color of an arc $\alpha\in \A_D$ of $D$ and $\ang{S}$ is the subgroup of $G$ generated by the subset $S$.
Indeed, it is easy to verify that none of the moves changes this group.
The map $\Psi$ fits into the commutative diagram
\[
\xymatrix{
\{\text{$G$-colored diagrams}\}/\text{(1)--(6)} \ar[r]^-\Psi \ar@{->>}[d] & \SS_G^\fg, \\
\coprod_\Gamma \Map_0(\R^3\setminus \Gamma, X_G)/{\sim} \ar[ru]_-\Phi &
}
\]
where the vertical surjective map exists by Lemma~\ref{lem:2-skeleton}.
More precisely, for a $G$-colored diagram $D$, let $\Gamma$ be the corresponding spatial graph and note that 
(by the restrictions that Definition~\ref{def:diagram} places on the system of colors)
we have a homomorphism $\pi_1(\R^3\setminus \Gamma)\to G$ by sending a loop wrapping around an arc $\alpha$ to $\col(\alpha)$ (see Figure~\ref{fig:Wirtinger}).
By Lemma~\ref{lem:2-skeleton}, the homomorphism corresponds to a continuous map $\R^3\setminus \Gamma\to X_G$.
This gives a map $\{\text{$G$-colored diagrams}\}\to \coprod_\Gamma \Map_0(\R^3\setminus \Gamma, X_G)$ which descends to the vertical map above since the local moves (1)--(6) are realized by the equivalence relation $\sim$.

\begin{theorem}
\label{thm:comb_classif}
The map $\Psi$ is a bijection.
\end{theorem}

Theorem~\ref{thm:top_classif} follows from Theorem~\ref{thm:comb_classif} and the above commutative diagram.
Indeed, Theorem~\ref{thm:comb_classif} implies that the vertical map is bijective, and thus $\Phi$ is also a bijection.

\begin{remark}
\label{rem:basepoint}
In the case where we do not specify the basepoint, the above diagram induces the commutative diagram
\[
\xymatrix{
\{\text{$G$-colored diagrams}\}/\text{(1)--(7)} \ar[r]^-\Psi \ar@{->>}[d] & \SS_G^\fg/\text{conjugacy}, \\
\coprod_\Gamma \Map(\R^3\setminus \Gamma, X_G)/{\sim} \ar[ru]_-\Phi &
}
\]
in which all three maps are bijections.
In this remark, $\sim$ stands for the equivalence relation that is defined just as $\sim$ in Definition~\ref{def:equivalence}, except that the ambient isotopy $h_t$ is not assumed to be basepoint-preserving any more.
\end{remark}

As a byproduct, in a somewhat indirect way, we also obtain the claim that the moves (1)--(6) of the previous section are (necessary and) sufficient to describe diagrammatically the equivalence relation $\sim$ of Section~\ref{sec:Mathematical model of global defects}.
The same can be said regarding the moves (1)--(7) and the basepoint-free version of $\sim$.
Indeed, this is exactly what the bijectivity of the `vertical maps' of the above commutative diagrams means.

\begin{remark}
\label{rem:S3}
In Theorems~\ref{thm:top_classif}, \ref{thm:comb_classif} and Remark~\ref{rem:basepoint}, $\R^3$ can be replaced with the closed $3$-ball $D^3$.
In the case where the ambient space is $S^3$ with basepoint $p_0$, similar results hold after replacing $\SS_G^\fg$ with $\SS_G^\fg\setminus\{\{1\}\}$ and excluding $G$-colored diagrams $D$ with $\Psi(D)=\{1\}$.
This is necessary because if $\Gamma$ is empty then 
$\Map_0(S^3,X_G)/{\simeq}=\pi_3(X_G)$ 
may be non-trivial, as is the case when $G=Q$ and $X_G=S^3/Q$ which has $\pi_3(S^3/Q)\cong\Z$.
Note also that the mapping
\[
\coprod_\Gamma \Map((D^3\setminus \Gamma,\partial D^3), (X_G,x_0))/{\sim} \to \coprod_\Gamma \Map_0(S^3\setminus \Gamma, X_G)/{\sim},
\]
defined by extending via the constant map, is a bijection.
Here the left-hand side describes textures that are uniform far from the origin. 
\end{remark}

As we have seen in Section~\ref{subsec:quaternion}, 
for the quaternion group $Q$, we have 
\[
\SS_Q^\fg = 
\SS_Q^\fg/\text{conjugacy} = 
\{ \{ 1 \},\ 
\langle -1 \rangle,\ 
\langle i \rangle,\ 
\langle j \rangle,\ 
\langle k \rangle,\ Q \} . 
\]
Thus, by 
Theorems~\ref{thm:top_classif}, \ref{thm:comb_classif} and 
Remark~\ref{rem:basepoint} above, 
we have the following. 
\begin{corollary}
\label{cor:classification biaxial}
The equivalence classes of non-trivial global defects in biaxial nematic 
liquid crystals in $\R^3$ or $S^3$ are represented by the five 
diagrams \textup{(}excluding $\emptyset$\textup{)}
shown in Figure~\ref{fig:classification_correspondence}.
%
%
\end{corollary}

\begin{example}
The three spatial graph diagrams depicted in Figure~\ref{fig:classification_example2} are all equivalent. 
Indeed, the sugroups of $Q$ that their colors generate are all equal to $\langle i \rangle$. 
Thus, by Corollary~\ref{cor:classification biaxial}, we see immediately that they are equivalent without actually relating them by moves.
Similarly, the four spatial graph diagrams depicted in Figure~\ref{fig:classification_example1} are all equivalent.
The subgroups of $Q$ generated by their colors are all equal to the 
whole group $Q$ in this case. 
In the appendix, an actual explicit
sequence of moves 
from a colored $(2,4)$-torus knot diagram to 
a colored theta-graph diagram will be given.

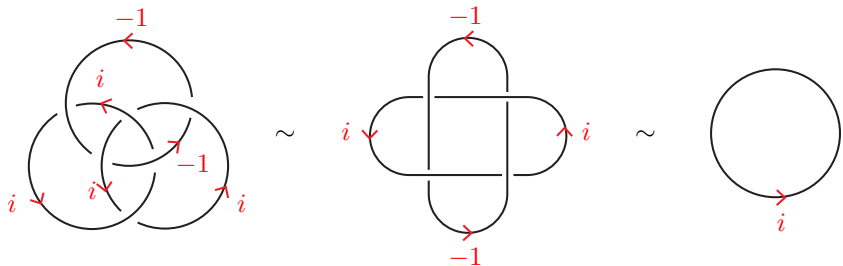
\begin{figure}[htbp]
\centering\includegraphics[height=2.8cm]{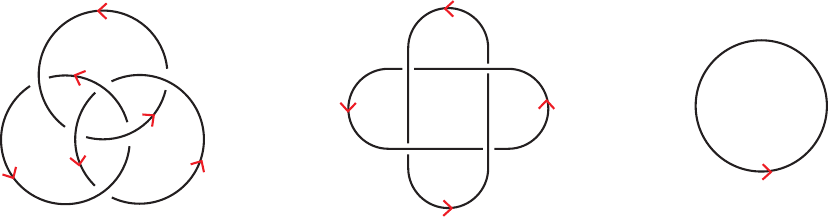}
\begin{picture}(400,0)(0,0)
\put(80,93){\color{red} $-1$}
\put(40,23){\color{red} $i$}
\put(126,23){\color{red} $i$}
\put(103,38){\color{red} $-1$}
\put(73,70){\color{red} $i$}
\put(70,28){\color{red} $i$}

\put(140,50){$\sim$}

\put(205,94){\color{red} $-1$}
\put(205,3){\color{red} $-1$}
\put(165,50){\color{red} $i$}
\put(255,50){\color{red} $i$}

\put(275,50){$\sim$}

\put(328,16){\color{red} $i$}
\end{picture}
\caption{Equivalent defects corresponding to (the conjugacy class of) the subgroup 
$\langle i \rangle$ of $Q$.}
\label{fig:classification_example2}
\end{figure}

\begin{figure}[htbp]
\centering\includegraphics[height=2.8cm]{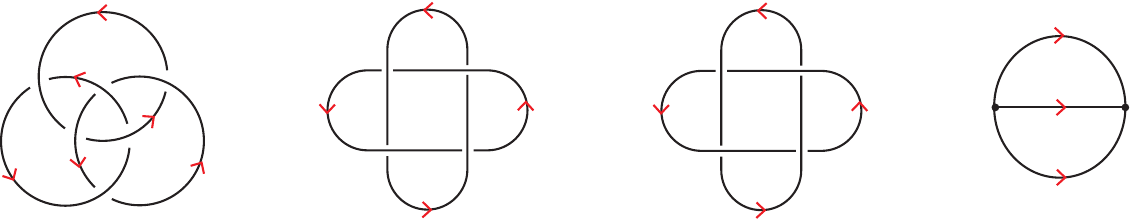}
\begin{picture}(400,0)(0,0)
\put(30,93){\color{red} $i$}
\put(-15,23){\color{red} $j$}
\put(71,23){\color{red} $k$}
\put(48,38){\color{red} $-i$}
\put(18,70){\color{red} $-j$}
\put(3,28){\color{red} $-k$}

\put(86,50){$\sim$}

\put(150,93){\color{red} $i$}
\put(145,3){\color{red} $-i$}
\put(104,50){\color{red} $j$}
\put(191,50){\color{red} $-j$}

\put(210,50){$\sim$}

\put(271,93){\color{red} $k$}
\put(262,3){\color{red} $-k$}
\put(225,50){\color{red} $j$}
\put(313,50){\color{red} $-j$}

\put(335,50){$\sim$}

\put(378,85){\color{red} $i$}
\put(378,60){\color{red} $j$}
\put(378,15){\color{red} $k$}
\end{picture}
\caption{Equivalent defects corresponding to (the conjugacy class of) the subgroup $Q$ of $Q$.}
\label{fig:classification_example1}
\end{figure}
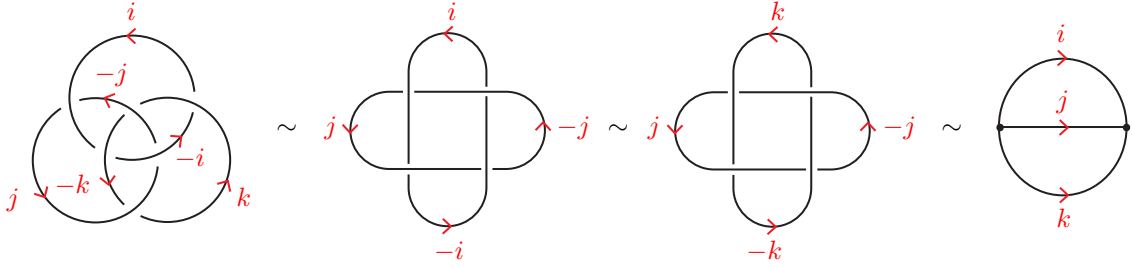
\end{example}

\begin{remark}
It is natural to ask whether a defect can be deformed into a knot, or equivalently, for a given finitely generated subgroup $H$ of $G$, whether there exist a knot $K$ and continuous map $f\in \Map_0(\R^3\setminus K, X_G)$ satisfying $\Phi([f])=H$. 
If it happens, then $H$ should be normally generated by one element, and the converse is also true due to \cite{Gon75}.
In the case of $Q$, for the subgroups $\ang{-1}$, $\ang{i}$, $\ang{j}$, and $\ang{k}$, we already find such knots in Figure~\ref{fig:classification_correspondence}; on the other hand there is no such knot for $H=\ang{i,j}$ since it is not normally generated by one element. 
More generally, 
by a result of Gonz\'alez-Acu\~na \cite{Gon75}, 
$H$ is normally generated by $r$ elements if and only if there exist an $r$-component link $L$ and $f\in \Map_0(\R^3\setminus L, X_G)$ satisfying $\Phi([f])=H$.
\end{remark}

\appendix
\section{Proof of Theorem~\ref{thm:comb_classif}}

An $n$-\textit{rose} is a graph consisting of a single 
vertex and $n$ loops. 
A diagram of a spatial $n$-rose with no crossings is called 
the \textit{standard diagram} of a spatial $n$-rose. 
In order to prove Theorem~\ref{thm:comb_classif}, 
we first establish three lemmas regarding equivalences of standard 
diagrams of spatial roses colored by elements of a group $G$. 

\begin{lemma}[Inversion]
\label{lemma:Inversion}
If a colored standard diagram of a spatial rose is obtained from 
another by reversing the orientation of just one edge, then 
they are equivalent, 
that is, they can be related by a finite sequence of moves \textup{(1)--(6)} introduced in Section~\ref{sec:Colored spatial graph diagrams}.
%
%
\end{lemma}

\begin{proof}
Figure~\ref{fig:inversion_proof} shows 
a sequence of moves between the two diagrams. 
\begin{figure}[htbp]
\centering\includegraphics[width=14cm]{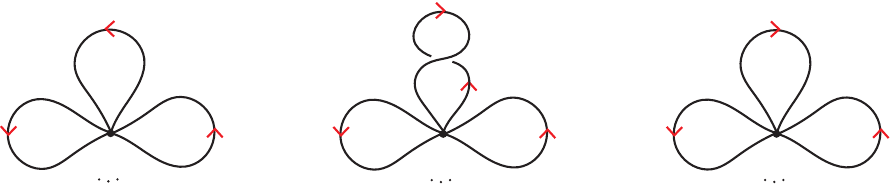}
\begin{picture}(400,0)(0,0)
\put(47,90){\color{red} $c_1$}
\put(-10,33){\color{red} $c_2$}
\put(103,33){\color{red} $c_n$}

\put(115,40){$\xrightarrow{(2)\RR_1}$}

\put(194,97){\color{red} $c_1$}
\put(139,33){\color{red} $c_2$}
\put(252,33){\color{red} $c_n$}

\put(265,40){$\xrightarrow{(2)\RR_5}$}

\put(345,90){\color{red} $c_1$}
\put(288,33){\color{red} $c_2$}
\put(401,33){\color{red} $c_n$}
\end{picture}
\caption{A sequence of moves that realizes an inversion.}
\label{fig:inversion_proof}
\end{figure}
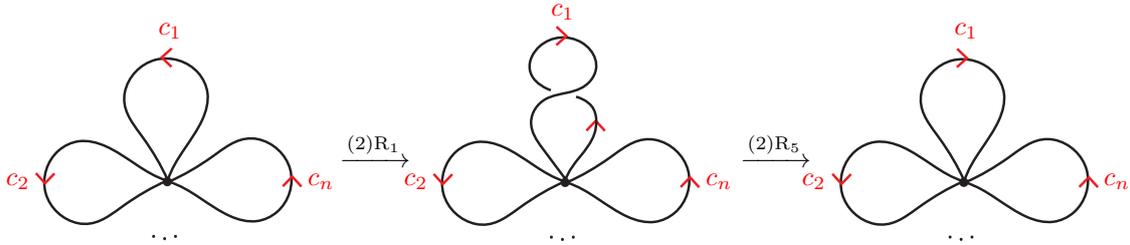
\end{proof}

\begin{lemma}[Switch]
If a colored standard diagram of spatial rose is obtained from 
another by just switching the colors of two adjacent edges, then 
they are equivalent. 
%
%
\end{lemma}

\begin{proof}
Figure \ref{fig:switch_proof} shows 
a sequence of moves between the two diagrams. 
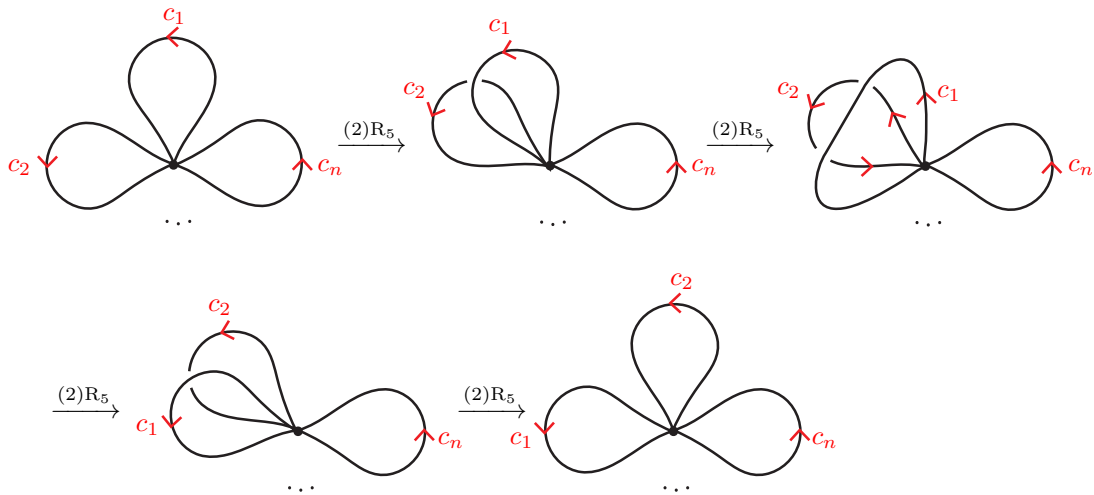
\begin{figure}[h!]
\centering\includegraphics[width=13.5cm]{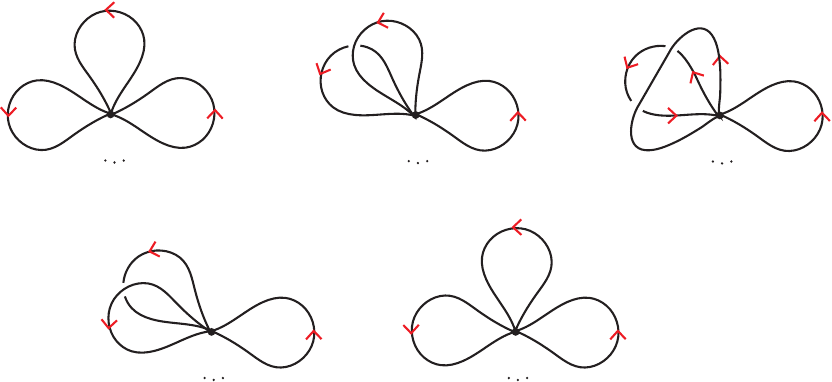}
\begin{picture}(400,0)(0,0)

\put(55,190){\color{red} $c_1$}
\put(-3,133){\color{red} $c_2$}
\put(113,133){\color{red} $c_n$}

\put(120,140){$\xrightarrow{(2)\RR_5}$}

\put(177,186){\color{red} $c_1$}
\put(145,162){\color{red} $c_2$}
\put(253,133){\color{red} $c_n$}

\put(258,140){$\xrightarrow{(2)\RR_5}$}

\put(345,161){\color{red} $c_1$}
\put(285,162){\color{red} $c_2$}
\put(394,133){\color{red} $c_n$}

\put(13,40){$\xrightarrow{(2)\RR_5}$}

\put(72,80){\color{red} $c_2$}
\put(45,35){\color{red} $c_1$}
\put(158,30){\color{red} $c_n$}

\put(165,40){$\xrightarrow{(2)\RR_5}$}

\put(245,90){\color{red} $c_2$}
\put(185,30){\color{red} $c_1$}
\put(298,30){\color{red} $c_n$}
\end{picture}
\caption{A sequence of moves that realizes a switch.}
\label{fig:switch_proof}
\end{figure}
\end{proof}

\begin{lemma}[Transvection]
\label{lemma:Transvection}
The two $G$-colored spatial graph diagrams shown in 
Figure~\ref{fig:transvection} are equivalent.
\begin{figure}[h!]
\centering\includegraphics[width=8cm]{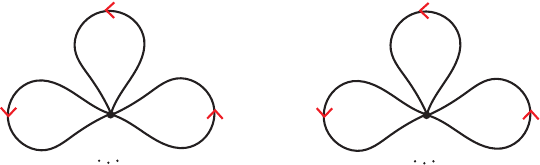}
\begin{picture}(400,0)(0,0)
\put(130,85){\color{red} $c_1$}
\put(76,33){\color{red} $c_2$}
\put(182,33){\color{red} $c_n$}

\put(196,35){$\sim$}

\put(258,85){\color{red} $c_1 c_2$}
\put(209,33){\color{red} $c_2$}
\put(315,33){\color{red} $c_n$}
\end{picture}
\caption{Transvection.}
\label{fig:transvection}
\end{figure}
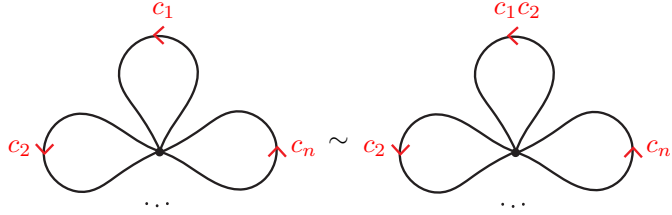
\end{lemma}

\begin{proof}
Figure~\ref{fig:transvection_proof} shows 
a sequence of moves between the two diagrams. 
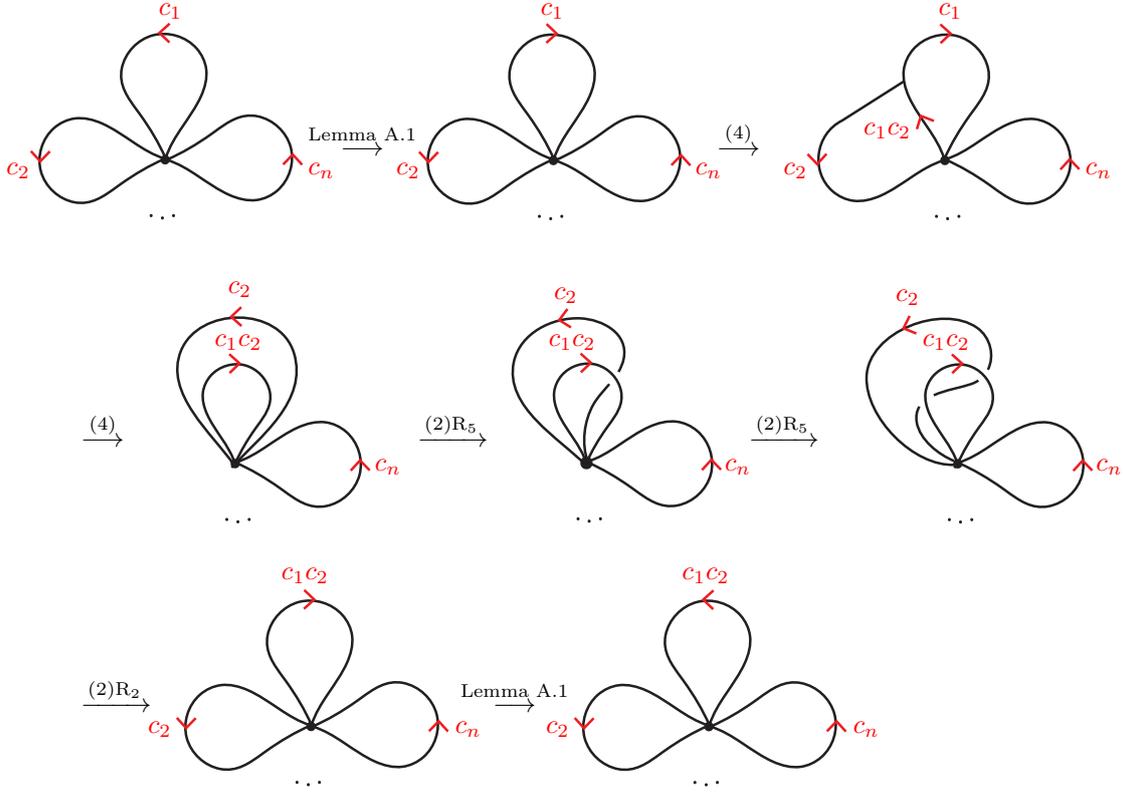
\begin{figure}[htbp]
\centering\includegraphics[width=14cm]{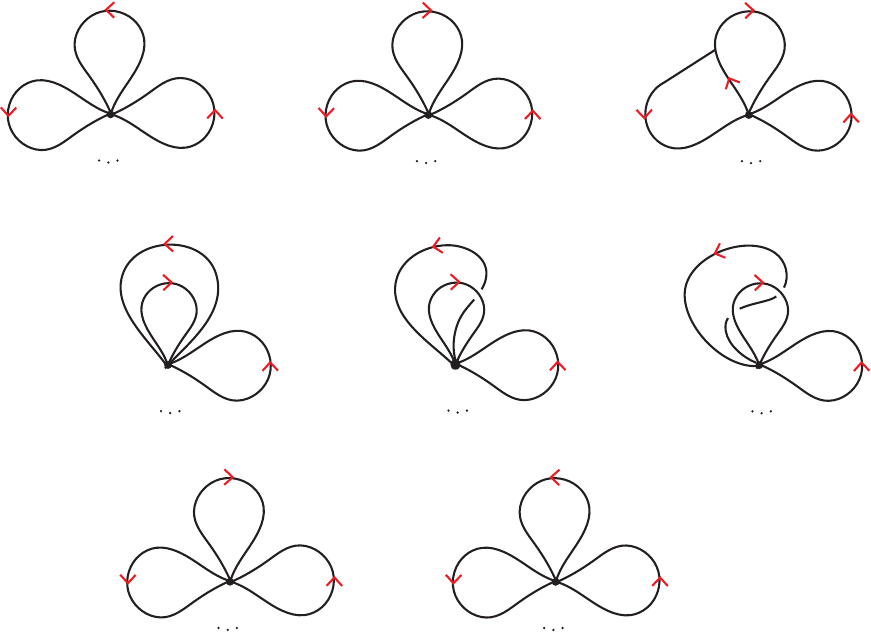}
\begin{picture}(400,0)(0,0)
\put(49,303){\color{red} $c_1$}
\put(-8,243){\color{red} $c_2$}
\put(105,243){\color{red} $c_n$}

\put(105,250){$\overset{\text{\scriptsize Lemma~\ref{lemma:Inversion}}}{\longrightarrow}$}

\put(192,303){\color{red} $c_1$}
\put(138,243){\color{red} $c_2$}
\put(250,243){\color{red} $c_n$}

\put(258,250){$\overset{(4)}{\longrightarrow}$}

\put(341,303){\color{red} $c_1$}
\put(283,243){\color{red} $c_2$}
\put(313,258){\color{red} $c_1 c_2$}
\put(396,243){\color{red} $c_n$}

\put(20,140){$\overset{(4)}{\longrightarrow}$}

\put(75,198){\color{red} $c_2$}
\put(70,178){\color{red} $c_1 c_2$}
\put(130,131){\color{red} $c_n$}

\put(146,140){$\xrightarrow{(2)\RR_5}$}

\put(197,196){\color{red} $c_2$}
\put(195,178){\color{red} $c_1 c_2$}
\put(261,131){\color{red} $c_n$}

\put(270,140){$\xrightarrow{(2)\RR_5}$}

\put(325,195){\color{red} $c_2$}
\put(335,178){\color{red} $c_1 c_2$}
\put(400,131){\color{red} $c_n$}

\put(20,40){$\xrightarrow{(2)\RR_2}$}

\put(95,90){\color{red} $c_1 c_2$}
\put(45,32){\color{red} $c_2$}
\put(160,32){\color{red} $c_n$}

\put(162,40){$\overset{\text{\scriptsize Lemma~\ref{lemma:Inversion}}}{\longrightarrow}$}

\put(245,90){\color{red} $c_1 c_2$}
\put(194,32){\color{red} $c_2$}
\put(309,32){\color{red} $c_n$}
\end{picture}
\caption{A sequence of moves that realizes a transvection.}
\label{fig:transvection_proof}
\end{figure}
\end{proof}

\begin{proof}[Proof of Theorem~\ref{thm:comb_classif}]
The surjectivity of $\Psi$ is clear since, for a subgroup $H$ of $G$ generated by $c_1, c_2, \ldots, c_r$, the standard diagram of the $r$-rose with colors $c_1, c_2, \ldots, c_r$ on its edges represents an element in the preimage $\Psi^{-1}(H)$. 

It remains to show that $\Psi$ is injective. 
Let $H$ be a finitely generated subgroup of $G$. 
Let $D$ be an arbitrary $G$-colored spatial graph diagram in the preimage $\Psi^{-1}(H)$. 
We choose and fix generators $c_1,c_2,\ldots,c_r$ of $H$. 
It suffices to show that $D$ is equivalent to the standard 
diagram of the $r$-rose with colors 
$c_1, c_2, \ldots, c_r$ on its edges. 
At each crossing of $D$, 
perform the sequence of moves shown in 
Figure~\ref{fig:unknotting_tunnel}. 
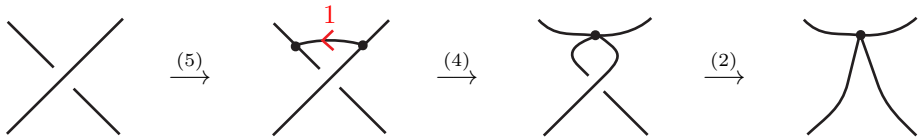
\begin{figure}[htbp]
\centering\includegraphics[width=12cm]{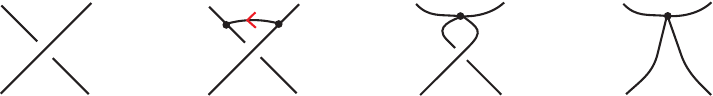}
\begin{picture}(400,0)(0,0)
\put(90,32){$\overset{(5)}{\longrightarrow}$} 
\put(190,32){$\overset{(4)}{\longrightarrow}$} 
\put(290,32){$\overset{(2)}{\longrightarrow}$} 
\put(148,55){\color{red} $1$} 
\end{picture}
\caption{A sequence of moves that changes a crossing point to a vertex.}
\label{fig:unknotting_tunnel}
\end{figure}
The resulting diagram $D'$ is then a graph without 
crossings. 
By applying edge-contractions successively along a maximal tree of the diagram $D'$, 
we get a $G$-colored standard 
diagram $D''$ of the $n$-rose for some 
non-negative integer $n$. 
Let $b_1 , \ldots , b_n$ be the colors of 
the edges of $D''$. 
Now using the moves (4) and (5) appropriately, 
we can easily add one more loop, with the color 
$1$, to $D''$. 
Since the element $c_1 \in H$ can be 
expressed as a product of 
$b_1^{\pm1}, \ldots , b_m^{\pm1}$, 
we can change the color $1$ of the loop 
to $c_1$ by applying 
Lemmas~\ref{lemma:Inversion}--\ref{lemma:Transvection}, see Figure~\ref{fig:change_of_generators}. 
\begin{figure}[h]
\centering\includegraphics[width=14cm]{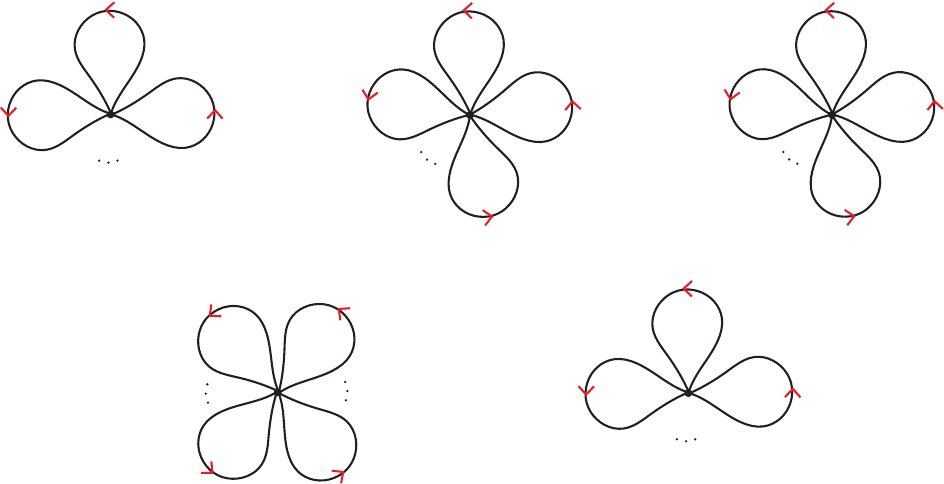}
\begin{picture}(400,0)(0,0)
\put(43,219){\color{red} $b_1$}
\put(-12,165){\color{red} $b_2$}
\put(97,165){\color{red} $b_n$}

\put(108,165){$\xrightarrow{(4),\ (5)}$}

\put(194,219){\color{red} $b_1$}
\put(142,174){\color{red} $b_2$}
\put(202,112){\color{red} $b_n$}
\put(248,169){\color{red} $1$}

\put(267,165){$\longrightarrow$}

\put(347,219){\color{red} $b_1$}
\put(296,174){\color{red} $b_2$}
\put(354,112){\color{red} $b_n$}
\put(401,169){\color{red} $c_1$}

\put(5,49){$\longrightarrow ~ \cdots ~ \longrightarrow$}

\put(78,85){\color{red} $b_1$}
\put(78,8){\color{red} $b_n$}
\put(151,85){\color{red} $c_1$}
\put(151,8){\color{red} $c_r$}

\put(166,49){$\longrightarrow ~ \cdots ~ \longrightarrow$}

\put(288,102){\color{red} $c_1$}
\put(234,49){\color{red} $c_2$}
\put(341,49){\color{red} $c_r$}
\end{picture}
\caption{A sequence of moves that relates the standard diagram of a rose with colors $b_1 , \ldots , b_n$ to that with $c_1 , \ldots , c_r$.}
\label{fig:change_of_generators}
\end{figure}
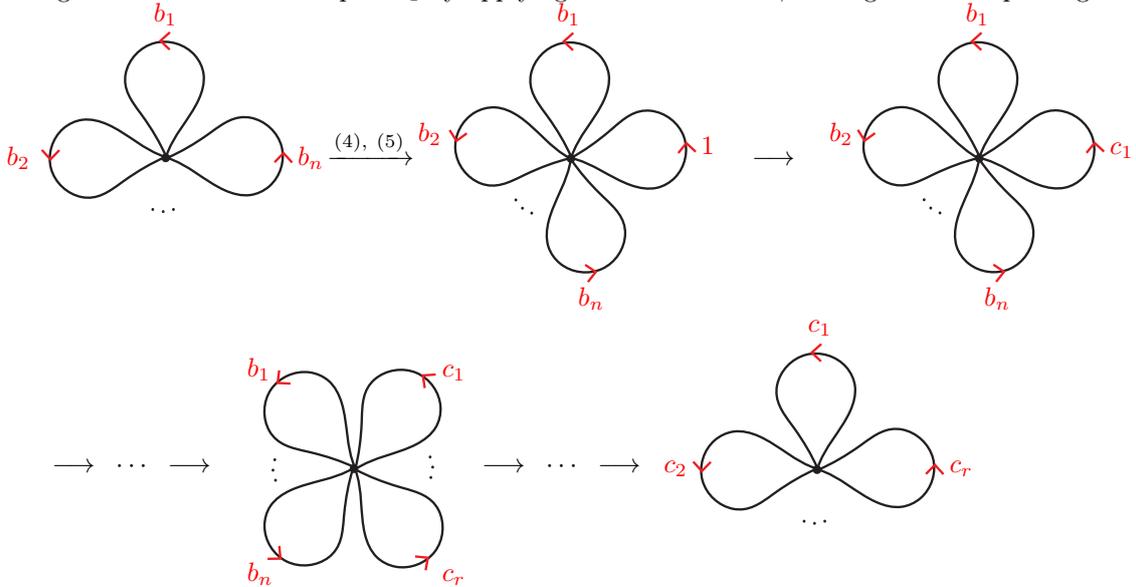
Repeating this process $r$ times, 
we get the standard diagram of the $(n+r)$-rose with colors 
$b_1, \ldots, b_n, c_1 , \dots, c_r$. 
Then, by performing the inverse of this process 
reversing the roles of the $b_i$'s and the $c_j$'s, 
we finally get the standard 
diagram of the $r$-rose with colors $c_1 , \dots, c_r$, as desired. 
\end{proof}

Figure~\ref{fig:24torus_link} shows a sequence of moves 
from a colored $(2,4)$-torus knot diagram to 
a colored theta-graph diagram according to the flow of the above proof. 

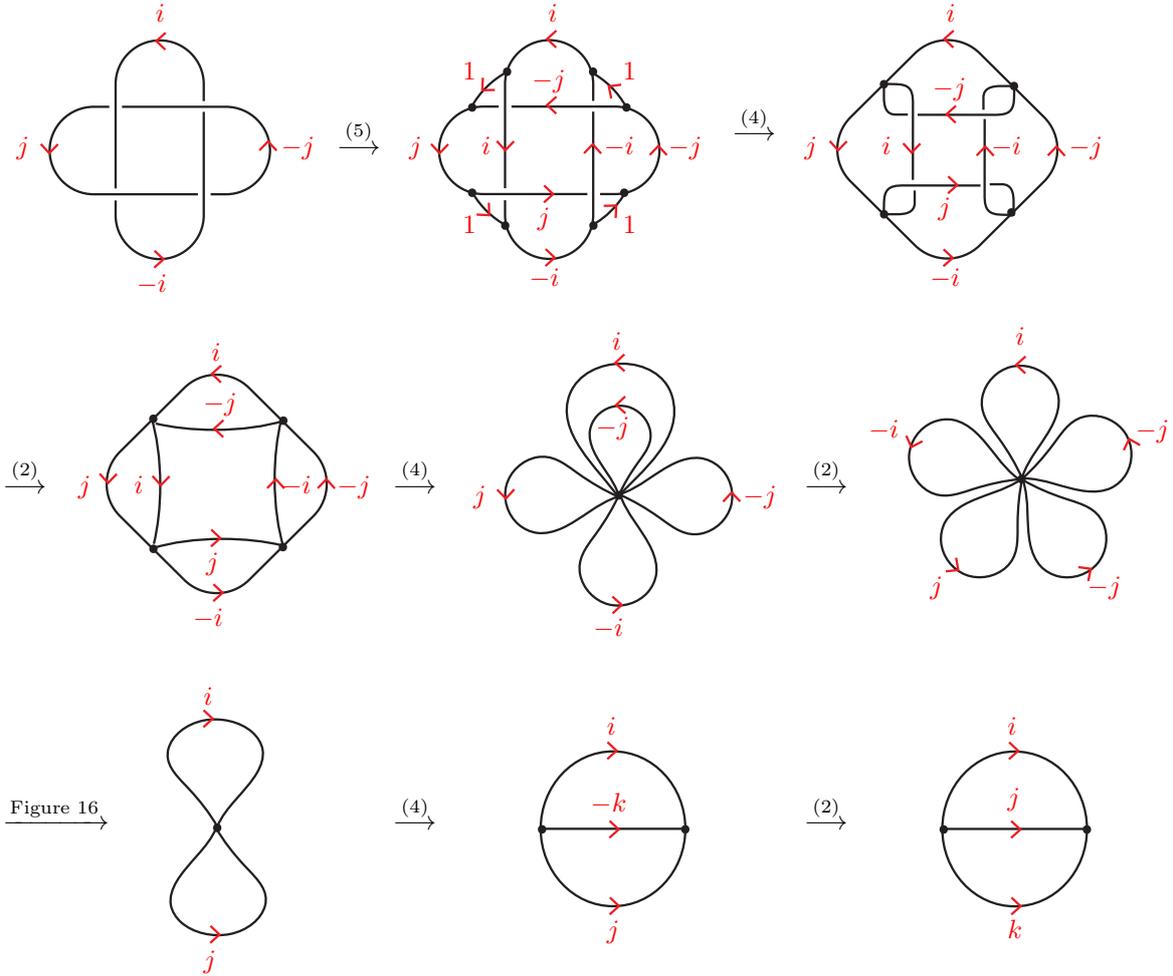
\begin{figure}[h]
\centering\includegraphics[width=14.5cm]{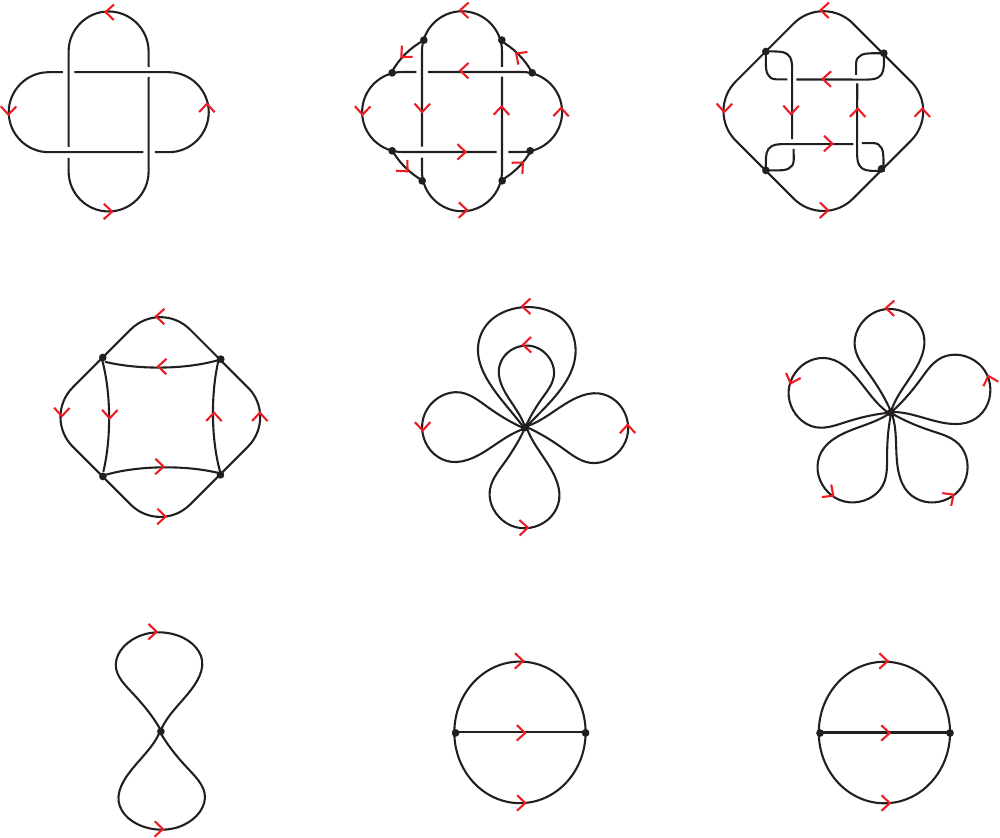}
\begin{picture}(400,0)(0,0)
\put(37,360){\color{red}$i$}
\put(30,258){\color{red}$-i$}
\put(-15,310){\color{red}$j$}
\put(84,310){\color{red}$-j$}

\put(105,310){$\overset{(5)}{\longrightarrow}$}

\put(184,360){\color{red}$i$}
\put(178,335){\color{red}$-j$}
\put(177,260){\color{red}$-i$}
\put(180,283){\color{red}$j$}
\put(132,310){\color{red}$j$}
\put(159,310){\color{red}$i$}
\put(229,310){\color{red}$-j$}
\put(205,310){\color{red}$-i$}
\put(152,338){\color{red}$1$}
\put(152,280){\color{red}$1$}
\put(212,338){\color{red}$1$}
\put(212,280){\color{red}$1$}

\put(253,315){$\overset{(4)}{\longrightarrow}$}

\put(333,360){\color{red}$i$}
\put(328,332){\color{red}$-j$}
\put(327,260){\color{red}$-i$}
\put(330,287){\color{red}$j$}
\put(280,310){\color{red}$j$}
\put(309,310){\color{red}$i$}
\put(379,310){\color{red}$-j$}
\put(350,310){\color{red}$-i$}

\put(-20,182){$\overset{(2)}{\longrightarrow}$}

\put(58,232){\color{red}$i$}
\put(55,213){\color{red}$-j$}
\put(51,132){\color{red}$-i$}
\put(56,153){\color{red}$j$}
\put(8,182){\color{red}$j$}
\put(29,182){\color{red}$i$}
\put(105,182){\color{red}$-j$}
\put(84,182){\color{red}$-i$}

\put(126,182){$\overset{(4)}{\longrightarrow}$}

\put(208,236){\color{red}$i$}
\put(202,204){\color{red}$-j$}
\put(201,128){\color{red}$-i$}
\put(156,178){\color{red}$j$}
\put(257,178){\color{red}$-j$}

\put(280,182){$\overset{(2)}{\longrightarrow}$}

\put(359,238){\color{red}$i$}
\put(304,203){\color{red}$-i$}
\put(404,203){\color{red}$-j$}
\put(327,144){\color{red}$j$}
\put(386,144){\color{red}$-j$}

\put(-20,55){$\xrightarrow{\text{Figure~\ref{fig:change_of_generators}}}$}

\put(55,102){\color{red}$i$}
\put(55,3){\color{red}$j$}

\put(126,55){$\overset{(4)}{\longrightarrow}$}

\put(206,91){\color{red}$i$}
\put(200,62){\color{red}$-k$}
\put(206,14){\color{red}$j$}

\put(280,55){$\overset{(2)}{\longrightarrow}$}

\put(356,91){\color{red}$i$}
\put(356,64){\color{red}$j$}
\put(356,14){\color{red}$k$}
\end{picture}
\caption{A sequence of moves from a $Q$-colored $(2,4)$-torus knot diagram to an (unknotted) $Q$-colored theta-graph diagram.}
\label{fig:24torus_link}
\end{figure}

\section{Tetrahedral nematic liquid crystals}
\label{sec:Tetrahedral}
As one more specific example in addition to biaxial nematic liquid crystals, 
we present here the classification of global defects in
local tetrahedral order, 
discussed in \cite{Tre84} and \cite{RoSe22}, among others. 
The fundamental group of such an order parameter space is the \emph{tetrahedral group} $T$, 
which is a finite group consisting of 12 elements, $1, i, j, k, \alpha^{\pm 1}, \beta^{\pm 1} , \gamma^{\pm 1} , 
\delta^{\pm 1}$, where the non-trivial elements are described in Figure~\ref{fig:tetrahedral_group}. 
\begin{figure}[h!]
\centering\includegraphics[width=14cm]{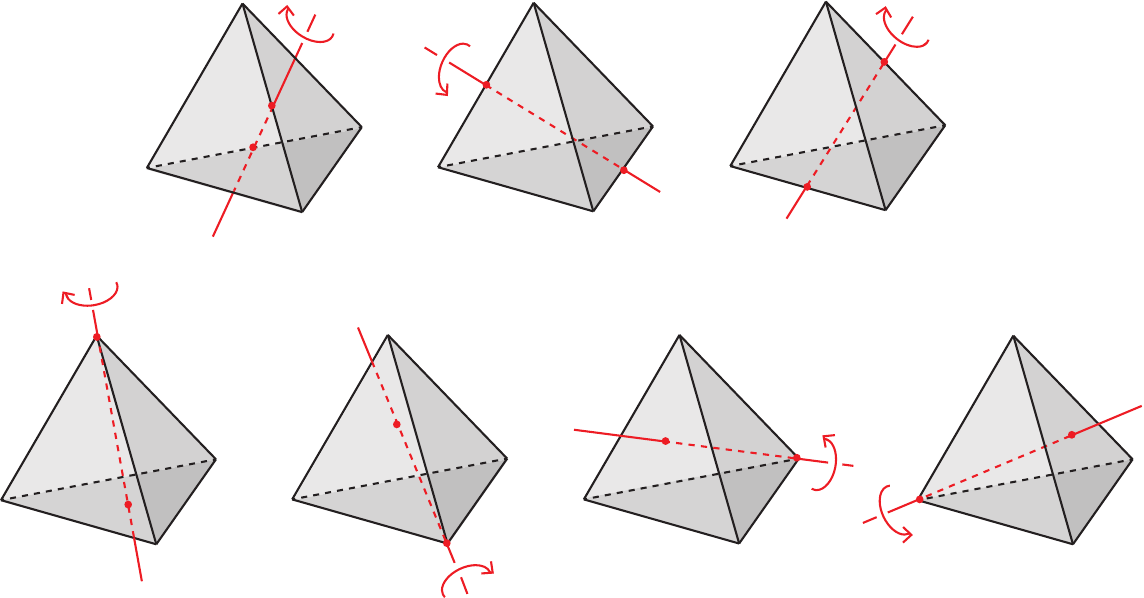}
\begin{picture}(400,0)(0,0)
\put(83,222){$1$}
\put(106,133){$2$}
\put(131,173){$3$}
\put(43,155){$4$}
\put(110,218){\color{red} $\pi$}
\put(87,125){$i$}

\put(185,222){$1$}
\put(208,133){$2$}
\put(233,173){$3$}
\put(145,155){$4$}
\put(141,203){\color{red} $\pi$}
\put(189,125){$j$}

\put(287,222){$1$}
\put(310,133){$2$}
\put(335,173){$3$}
\put(247,155){$4$}
\put(319,218){\color{red} $\pi$}
\put(291,125){$k$}

\put(38,104){$1$}
\put(55,19){$2$}
\put(80,57){$3$}
\put(-7,40){$4$}
\put(27,127){\color{red} $\frac{2\pi}{3}$}
\put(38,5){$\alpha$}

\put(134,107){$1$}
\put(151,20){$2$}
\put(182,57){$3$}
\put(94,42){$4$}
\put(163,0){\color{red} $\frac{2\pi}{3}$}
\put(140,5){$\beta$}

\put(236,107){$1$}
\put(258,19){$2$}
\put(280,62){$3$}
\put(196,42){$4$}
\put(300,56){\color{red} $\frac{2\pi}{3}$}
\put(242,5){$\gamma$}

\put(353,107){$1$}
\put(372,19){$2$}
\put(399,57){$3$}
\put(316,36){$4$}
\put(287,33){\color{red} $\frac{2\pi}{3}$}
\put(355,5){$\delta$}

\end{picture}
\caption{The three $\pi$-rotations $i$, $j$, $k$, and 
the four $2\pi/ 3$-rotations $\alpha, \beta, \gamma, \delta$.}
\label{fig:tetrahedral_group}
\end{figure}
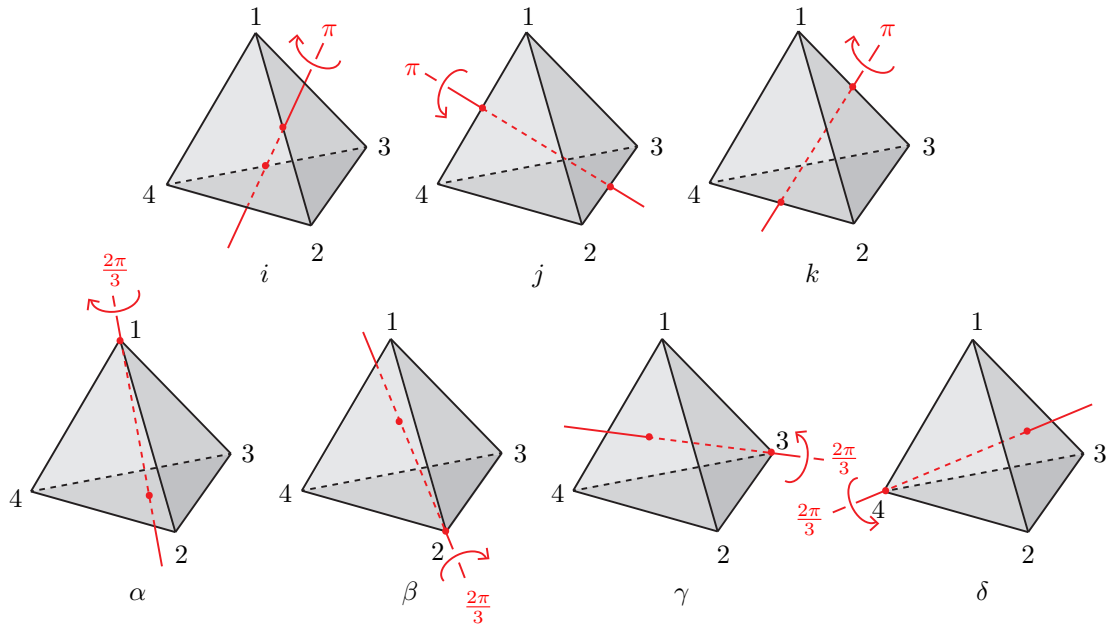

Note that the group $T$ is isomorphic to the 
alternating group $A_4$. 
The following is the list of all subgroups of $T$. 
\begin{enumerate}[label=(\arabic*)]
\item \textbf{(order $1$)}  $\{ 1 \} = \langle 1 \rangle$. 
\item \textbf{(order $2$)} $\langle i \rangle  \cong \Z / 2 \Z$, $\langle j \rangle$, $\langle k \rangle$.
\item \textbf{(order $3$)} $\langle \alpha \rangle  \cong \Z / 3\Z$, $\langle \beta \rangle$, $\langle \gamma \rangle$, $\langle \delta
\rangle$.
\item \textbf{(order $4$)} $
\langle i, j \rangle = \langle j, k \rangle = 
\langle i, k \rangle = \langle i, j, k \rangle 
\cong \Z_2 \times \Z_2$.
\item \textbf{(order $12$)} $\langle i, \alpha \rangle = \langle i, \beta \rangle = \langle i, \gamma \rangle = \langle i, \delta \rangle = \langle j, \alpha \rangle = \cdots = \langle k, \delta \rangle = T$.
\end{enumerate}

The subgroups within each item of the above list are conjugate to each other. 
Thus, $\langle i, j, k \rangle$ is the only non-trivial normal subgroup of $T$, and there are in total five conjugacy classes of subgroups. 
Using the symbols introduced in Section \ref{sec:Mathematical model of global defects}, 
this can be summarized as 
\[
\SS_T^\fg/\text{conjugacy} = 
\{ \{ 1 \},\ 
\langle i \rangle,\ 
\langle \alpha \rangle,\ 
\langle i, j, k \rangle,\ 
 T \} . 
\]
Therefore, up to basepoint-preserving homotopy, there are exactly nine non-trivial global defects in tetrahedral nematic liquid crystals, 
corresponding to the nine subgroups of $T$ other than $\{1\}$. 
Furthermore, up to the basepoint-free version of homotopy, the non-trivial global defects in this ordered media are classified into only four classes, with their representative colored diagrams shown in Figure~\ref{fig:classification_correspondence_tetrahedral_group}. 
\begin{figure}[htbp]
\centering\includegraphics[width=10cm]{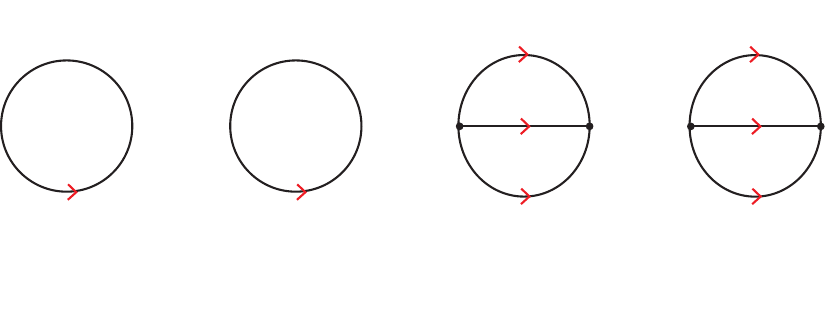}
\begin{picture}(400,0)(0,0)
\put(80,41){\color{red} $i$}
\put(158,41){\color{red} $\alpha$}
\put(235,107){\color{red} $i$}
\put(235,83){\color{red} $j$}
\put(235,39){\color{red} $k$}
\put(315,107){\color{red} $i$}
\put(315,83){\color{red} $\alpha$}
\put(314,39){\color{red} $\delta^2$}

\put(80,21){$\updownarrow$}
\put(158,21){$\updownarrow$}
\put(235,21){$\updownarrow$}
\put(315,21){$\updownarrow$}

\put(78,0) {$\langle i \rangle$}
\put(155,0){$\langle \alpha \rangle$}
\put(225,0){$\langle i, j, k \rangle$}
\put(315,0){$T$}
\end{picture}
\caption{Representatives of each
class 
of 
global defects in tetrahedral nematic liquid crystals in $S^3$, up to equivalence. Each corresponds to a conjugacy class of subgroups of the tetrahedral group $T$.}
\label{fig:classification_correspondence_tetrahedral_group}
\end{figure}
Note that, when we considered biaxial nematic liquid crystals, there was no difference between the classification up to basepoint-preserving homotopy and the one corresponding to (free) homotopy, because all subgroups of the quaternion group $Q$ are normal. However, 
for tetrahedral nematic liquid crystals, 
there is a substantial difference between the two classifications.

\end{document}